\documentclass[nonacm]{acmart}

\AtBeginDocument{%
  }

\setcopyright{acmcopyright}
\copyrightyear{2022}
\acmYear{2022}
\acmDOI{XXXXXXX.XXXXXXX}

\acmJournal{TCPS}
\acmVolume{37}
\acmNumber{4}
\acmArticle{1}
\acmMonth{7}

\usepackage{graphicx}

\usepackage{subfigure}

\usepackage{multicol}
\usepackage{tabu}

\usepackage{adjustbox}

\usepackage{multirow}
\usepackage{flushend}

\usepackage{float}

\usepackage{caption}

\usepackage{calc}

\usepackage{cancel}

\usepackage{mathrsfs}
\usepackage{algorithm}
\usepackage{algorithmic}

\usepackage{cuted}
\usepackage{blindtext, graphicx}
\usepackage{amsmath,float}
\usepackage{dirtytalk}
\usepackage{url}
\usepackage{amsmath,amsfonts,amsthm,bm}
\usepackage[utf8]{inputenc}
\usepackage[english]{babel}
\usepackage{suffix}
\usepackage{mathtools}
\usepackage{multicol}
\usepackage{tabu}

\usepackage{caption}

\usepackage{graphics}
\usepackage{adjustbox}
\usepackage{lipsum}
\usepackage{stfloats}
\usepackage{multirow}

\usepackage{flushend}
\usepackage{tikz,colortbl}
\usepackage{zref-savepos}
\usepackage{pifont}

\usetikzlibrary{calc}

\usepackage{multirow}

\allowdisplaybreaks

\usepackage{soul}

\usepackage{color}

\newcounter{NoTableEntry}
\renewcommand*{\theNoTableEntry}{NTE-\the\value{NoTableEntry}}

\DeclarePairedDelimiter\abs{\lvert}{\rvert}%

\newcommand{\norm}[1]{\left\lVert#1\right\rVert}

\newtheorem{definition}{Definition}
\newtheorem{proposition}{Proposition}
\newtheorem{theorem}{Theorem}
\newtheorem{problem}{Problem}

\begin{document}

\title{PolyARBerNN: A Neural Network Guided Solver and Optimizer for Bounded Polynomial Inequalities
}

\author{Wael Fatnassi}
\email{wfatnass@uci.edu}
\orcid{1234-5678-9012}
\author{Yasser Shoukry}
\email{yshoukry@uci.edu}
\orcid{0000-0002-8224-8477}
\affiliation{%
  \institution{University of California, Irvine}
  \state{CA}
  \country{USA}
  \postcode{92697}
}








\renewcommand{\shortauthors}{W. Fatnassi and Y. Shoukry}

\begin{abstract}

Constraints solvers play a significant role in the analysis, synthesis, and formal verification of complex cyber-physical systems. In this paper, we study the problem of designing a scalable constraints solver for an important class of constraints named polynomial constraint inequalities (also known as nonlinear real arithmetic theory). In this paper, we introduce a solver named PolyARBerNN that uses convex polynomials as abstractions for highly nonlinears polynomials. Such abstractions were previously shown to be powerful to prune the search space and restrict the usage of sound and complete solvers to small search spaces. Compared with the previous efforts on using convex abstractions, PolyARBerNN provides three main contributions namely (i) a neural network guided abstraction refinement procedure that helps selecting the right abstraction out of a set of pre-defined abstractions, (ii) a Bernstein polynomial-based search space pruning mechanism that can be used to compute tight estimates of the polynomial maximum and minimum values which can be used as an additional abstraction of the polynomials, and (iii) an optimizer that transforms polynomial objective functions into polynomial constraints (on the gradient of the objective function) whose solutions are guaranteed to be close to the global optima. These enhancements together allowed the PolyARBerNN solver to solve complex instances and scales more favorably compared to the state-of-art nonlinear real arithmetic solvers while maintaining the soundness and completeness of the resulting solver. In particular, our test benches show that PolyARBerNN achieved 100X speedup compared with Z3 8.9, Yices 2.6, and PVS (a solver that uses Bernstein expansion to solve multivariate polynomial constraints) on a variety of standard test benches. Finally, we implemented an optimizer called PolyAROpt that uses PolyARBerNN to solve constrained polynomial optimization problems. Numerical results show that PolyAROpt is able to solve high-dimensional and high order polynomial optimization problems with higher speed compared to the built-in optimizer in the Z3 8.9 solver.

\end{abstract}

%


\maketitle

\section{Introduction}
Constraint solvers and optimizers have been used heavily in the design, synthesis, and verification of cyber-physical systems \cite{morecite1, morecite2, morecite3, morecite4, morecite5, morecite6, morecite7, morecite8, morecite9, morecite10, morecite11}. Examples includes verification of neural network controlled autonomous systems \cite{autosysverif}, formal verification of human-robot interaction in healthcare scenarios \cite{robotverif}, automated synthesis for distributed cyber-physical systems \cite{cpssynthesis}, design for cyber-physical systems under sensor attacks \cite{cpsdesign}, air traffic management of unmanned aircraft systems~\cite{airtraffic}, software verification for the next generation space-shuttle \cite{spaceshuttle}, and conflict detection for aircraft~\cite{verifaircraft}.

In this paper, we will focus on the class of general multivariate polynomial constraints (also known as nonlinear real arithmetic). Multivariate polynomial constraints 
appear naturally in the design, synthesis, and verification of these systems. It is not then surprising that the amount of attention given to this problem in the last decade, as evidenced by the amount of off-the-self solvers that are designed to solve feasibility and optimization problems over general multivariate polynomial constraints, including Z3~\cite{Z3}, Coq~\cite{coq}, Yices~\cite{yices}, PVS~\cite{pvsnasa}, Cplex, \cite{cplex}, CVXOPT \cite{cvxopt}, and Quadprog \cite{quadprog}. Regardless of their prevalence in several synthesis and verification problems, well-known algorithms---that are capable of solving a set of polynomial constraints---are shown to be doubly exponential~\cite{complexityproblem1}, placing a significant challenge to design efficient solvers for such problems.


Recently, neural networks (NNs) have shown impressive empirical results in approximating unknown functions. This observation motivated several researchers to ask how to use NNs to tame the complexity of NP-hard problems. Examples are the use of NNs to design scalable solvers for program synthesis~\cite{progsyn}, traveling salesman problem~\cite{travelsales}, first-order theorem proving~\cite{1thepro}, higher-order theorem proving~\cite{highthepro}, and Boolean satisfiability (SAT) problems~\cite{satsolve}. While several of these solvers sacrifice either soundness or correctness guarantees, we are interested in this paper on using such empirically powerful NNs to design a sound and complete solver for nonlinear real arithmetic.


In addition to NNs, polynomials constitute a rich class of functions for which several approximators have been studied. Two of the most famous approximators for polynomials are Taylor approximation and Bernstein polynomials. These two approximators have been successfully used in solvers like Coq and PVS \cite{coqtaylor, pvsnasa}. This opens the question of how to combine all those approximation techniques, i.e., NNs, Taylor, and Bernstein approximations, to come up with a scalable solver that can reason about general multivariate polynomial constraints. 
We introduce PolyARBerNN, a novel sound and complete solver for polynomial constraints that combines these three function approximators (NNs, Taylor, and Bernstein) to prune the search space and produce small enough instances in which existing sound and complete solvers (based on the well-known Cylindrical Algebraic Decomposition algorithm) can easily reason about. In general, we provide the following contributions:



\begin{itemize}
    \item We introduced a novel NN-guided abstraction refinement process in which a NN is used to guide the use of Taylor approximations to find a solution or prune the search space. We analyzed the theoretical characteristics of such a NN and provided empirical evidence on the generalizability of the trained NN in terms of its ability to guide the abstraction refinement process for unseen polynomials with various numbers of variables and orders.
    \item We complement the NN-guided abstraction refinement with a state-space pruning phase using Bernstein approximations that accelerates the process of removing portions of the state space in which the sign of the polynomial does not change.
    
    
    \item We validated our approach by first comparing the scalability of the proposed PolyARBerNN solver with respect to PVS, a library that uses Bernstein expansion to solve polynomial constraints. Second, we compared the execution times of the proposed tool with the latest versions of the state-of-the-art nonlinear arithmetic solvers, such as Z3 8.9, Yices 2.6 by varying the order, the number of variables, and the number of the polynomial constraints for instances when a solution exists and when a solution does not exist. We also compared the scalability of the solver against Z3 8.9 and Yices 2.9 on the problem of synthesizing a controller for a cyber-physical system.
    
    \item We proposed PolyAROpt, an optimizer that uses PolyARBerNN to solve constrained multivariate polynomial optimization problems. Our theoretical analysis shows that PolyAROpt is capable of providing solutions that are $\epsilon$ close to the global optima (for any  $\epsilon > 0$ chosen by the user). 
    Numerical results show that PolyAROpt solves high-dimensional and high-order optimization problems with high speed compared to the built-in optimizer in Z3 8.9 solver. We also validated the effectiveness of PolyAROpt on the problem of computing the reachable sets of polynomial dynamical systems.
\end{itemize}

\textbf{Related work:}
Cylindrical algebraic decomposition (CAD) was introduced by Collins \cite{collins} in 1975 and is considered to be the first algorithm 
to effectively solve general polynomial inequality constraints. Several improvements were introduced across the years to reduce the high time complexity of the CAD algorithm~\cite{Hong, McCalum, brown}. Although the CAD algorithm is sound and complete, it scales poorly with the number of polynomial constraints and their order. Other techniques to solve general polynomial inequality constraints include the use of transformations and approximations to scale the computations. For instance, the authors in \cite{pvsnasa} incorporated Bernstein polynomials in the Prototype Verification System (PVS) theorem prover; these developments are publicly available in the NASA PVS Library. The library uses the range enclosure propriety of Bernstein polynomials to solve quantified polynomial inequalities. However, the library is not complete for non-strict inequalities \cite{pvsnasa} and is not practical for higher dimensional polynomials.
%
%
Another line of work that is related to our work is the use of machine learning to solve combinatorial problems \cite{gcnn,satsolve}. In particular, the authors in~\cite{gcnn} proposed a graph convolutional neural network (GCNN) to learn heuristics that can accelerate mixed-integer linear programming (MILP) solvers.
Similarly, the NeuroSAT solver~\cite{satsolve} uses a message-passing neural network (MPNN) to solve Boolean SAT problems. 
The authors of~\cite{satsolve} showed that NeuroSAT generalizes to novel distributions after training only on random SAT problems. Nevertheless, NeuroSAT is not competitive with state-of-art SAT solvers and it does not have a correctness guarantee. 
 


\section{Problem Formulation}
\textbf{Notation:}
We use the symbols $\mathbb{N}$ and $\mathbb{R}$ to denote the set of natural and real numbers, respectively.
We denote by $x=\big(x_1,x_2,\cdots,x_n\big) \in \mathbb{R}^n$ the vector of real-valued variables, where $x_i \in \mathbb{R}$. We denote by $I_n (\underline{d}, \overline{d}) =\big[\underline{d}_1,\overline{d}_1\big] \times \cdots \times$ $\big[\underline{d}_n,\overline{d}_n\big] \subset \mathbb{R}^{n}$ the $n$-dimensional hyperrectangle where $\underline{d} = \left(\underline{d}_1, \cdots, \underline{d}_n\right)$ and $\overline{d} = \left(\overline{d}_1, \cdots, \overline{d}_n\right)$ are the lower and upper bounds of the hyperrectangle, respectively. 
For a real-valued vector $x =\big(x_1,x_2,\cdots,x_n\big)\in \mathbb{R}^n$ and an index-vector $K = \left(k_1, \cdots, k_n\right) \in \mathbb{N}^n$, we denote by $x^K \in \mathbb{R}$ the scalar $x^K = x_1^{k_1}\cdots x_n^{k_n}$. 
Given two multi-indices $K = \left(k_1, \cdots, k_n\right) \in \mathbb{N}^n$ and $L = \left(l_1, \cdots, l_n\right) \in \mathbb{N}^n$, we use the following notation throughout this paper: $K + L = \left(k_1+l_1, \cdots, k_n + l_n\right)$, ${L \choose K}={l_1 \choose k_1} \cdots {l_n \choose k_n}$, and $\sum\limits_{K \leq L}=\sum\limits_{k_1 \leq l_1}^{}\cdots \sum\limits_{k_n \leq l_n}$. 
A real-valued multivariate polynomial $p:\mathbb{R}^n \rightarrow \mathbb{R}$ is defined as:
\begin{align*}
 p(x_1, \ldots, x_n) & \;=\; \sum_{k_1 = 0}^{l_1 } \sum_{k_2 = 0}^{l_2} \ldots \sum_{k_n = 0}^{l_n} a_{(k_1,\ldots,k_n)} x_1^{k_1} x_2^{k_2} \ldots x_n^{k_n} 
 \;=\;\sum\limits_{K \leq L} a_K x^K,
\end{align*}
where $L = (l_1, l_2, \ldots, l_n)$ is the maximum degree of $x_i$ for all $i = 1, \ldots, n$. We denote by $a_p = (a_{(0,0,\ldots,0)}, \ldots, a_{(l_1, l_2, \ldots, l_n)})$ the vector of all the coefficients of polynomial $p$.
We denote the space of multivariate polynomials with coefficients in $\mathbb{R}$ by $\mathbb{R}[x_1,x_2,\cdots,x_n]$.  
Given a real-valued function $f:\mathbb{R}^n\rightarrow\mathbb{R}$, we denote by  $L^{-}_{0}\left(f\right)$ and $L^{+}_{0}\left(f\right)$ the zero sublevel and zero superlevel sets of f, i.e.,:
\begin{align*}
L^{-}_{0}\left(f\right)=\{x \in \mathbb{R}^n\big|f\big(x\big)\leq 0\}, \qquad L^{+}_{0}\left(f\right)=\{x \in \mathbb{R}^n \big|f\big(x\big)\geq 0\}.
\end{align*}
Finally, a function $f:\mathbb{R}^n\rightarrow\mathbb{R}^m$ is called Lipschitz continuous if there exists a positive real constant $\omega_f \ge 0$ such that, for all $x_1 \in \mathbb{R}^n$ and $x_2 \in \mathbb{R}^n$, the following holds:
$$ \norm{f(x_1) - f(x_2)} \le \omega_f \norm{x_1 - x_2}$$


\textbf{Main Problem:}
In this paper, we focus on two problems namely (Problem 1) the feasibility problems that involve multiple \textit{polynomial inequality constraints} with input ranges confined within closed hyperrectangles and (Problem 2) the constrained optimization problem which aims to maximize (or minimize) a polynomial objective function subject to other polynomial inequality constraints and input range constraints. 
%
\begin{problem}\label{prob1}
\begin{align*}
    \exists x \in I_n (\underline{d}, \overline{d}) \quad \text{such that:} \quad  &p_1\left(x_1,\cdots,x_n\right)~\leq~0 \qquad \qquad \qquad \qquad \qquad \qquad \qquad\\
    &\qquad \qquad \vdots \\
    &p_m\left(x_1,\cdots,x_n\right)~\leq~0
\end{align*}
\end{problem}
\noindent where $p_i\left(x\right)=p_i\left(x_1,\cdots,x_n\right) \in \mathbb{R}[x_1,x_2,\cdots,x_n]$ is a polynomial over variables  $x_1,\cdots,x_n$. Without loss of generality, $p_i\big(x\big)$ $\geq~0$ and $p_i\left(x\right)~=~0$ can be encoded using the constraints above.
Similarly, given a polynomial objective function $p(x) \in \mathbb{R}[x_1,x_2,\cdots,x_n]$, we define the optimization problem as:
\begin{problem}
\begin{align*}
\min\limits_{x \in I_n (\underline{d}, \overline{d})} ~ & p\left(x\right) ~~[\textbf{or} ~ \max\limits_{x \in I_n (\underline{d}, \overline{d})} ~ p\left(x\right)] \qquad\qquad\qquad\\
\text{subject to:} \qquad&
p_1\left(x_1,\cdots,x_n\right)~\leq~0, \\    
&\qquad \qquad \vdots \\
& p_m\left(x_1,\cdots,x_n\right)~\leq~0
\end{align*}

\end{problem}


\section{Convex Abstraction Refinement: Benefits and Drawbacks}
In this section, we overview our previously reported framework for using
convex abstraction refinement process introduced in~\cite{polyar} along with some drawbacks that motivate the need for the proposed framework.


\subsection{Overview of Convex Abstraction Refinement}
Sound and complete algorithms that solve Problem 1 are known to be doubly exponential in $n$ with a total running time that it is bounded by $\left(m~\overline{deg}\right)^{2^n}$ \cite{complexityproblem1}, where $\overline{deg}$ is the maximum degree among the polynomials $p_1, \ldots, p_m$. Since the complexity of the problem grows exponentially, it is useful to remove (or prune) subsets of the search space in which the solution is guaranteed not to exist. Since Problem 1 asks for an $x$ in $\mathbb{R}^n$ for which all the polynomials are negative, a solution does not exist in subsets of $\mathbb{R}^n$ at which one of the polynomials $p_i$ is always positive (i.e., $L^{+}_{0}\left(p_i\right)$). In the same way, finding regions of the input space for which some of the polynomials are negative $L^{-}_{0}\left(p_i\right)$ helps with finding the solution faster.

To find subsets of $L^{+}_{0}\left(p_i\right)$ and $L^{-}_{0}\left(p_i\right)$ efficiently, the use of ``convex abstractions'' of the polynomials was previously proposed by the authors in~\cite{polyar}. Starting from a polynomial $p_i(x) \in \mathbb{R}[x]$ and a hyperrectangle $I_n \subset \mathbb{R}^n$, the framework in~\cite{polyar} computes two quadratic polynomials $O^{p_i}_j$ and $U^{p_i}_j$ such that:
\begin{align}
    U^{p_i}_j(x) \le p(x) \le O^{p_i}_j(x), \qquad \forall x \in I_n,
\end{align}
where $O$ and $U$ stands for Over-approximate and Under-approximate quadratic polynomials, respectively, and the subscript $j$ in $O^{p_i}_j(x)$ and $U^{p_i}_j(x)$ encodes the iteration index of the abstraction refinement process. It is easy to notice that the zero superlevel set of $U^{p_i}_j(x)$ is a subset of $L^{+}_{0}(p_i)$, i.e., $L^{+}_{0}(U^{p_i}_j) \subseteq L^{+}_{0}(p_i)$. Similarly, the zero sublevel set of $O^{p_i}_j(x)$ is a subset of $L^{-}_{0}(p_i)$, i.e., $L^{-}_{0}(O^{p_i}_j) \subseteq L^{-}_{0}(p_i)$. Moreover, being convex polynomials, identifying the zero superlevel sets and zero sublevel sets of $O^{p_i}_j(x)$ and $U^{p_i}_j(x)$ can be computed efficiently using convex programming tools. By iteratively refining these upper and lower convex approximations, the framework in~\cite{polyar} was able to rapidly prune the search space until regions with relatively small volumes are identified, at which sound and complete tools such as Z3 8.9 and Yices 2.6 (which are based on the Cylindrical Algebraic Decomposition algorithm) are used to search these small regions, efficiently, to find a solution. It is important to notice that these solvers (especially Yices) are optimized for the cases when the search space is a bounded hyperrectangle.


\begin{figure*}[!ht]
    \centering
	\includegraphics[width=0.49\textwidth]{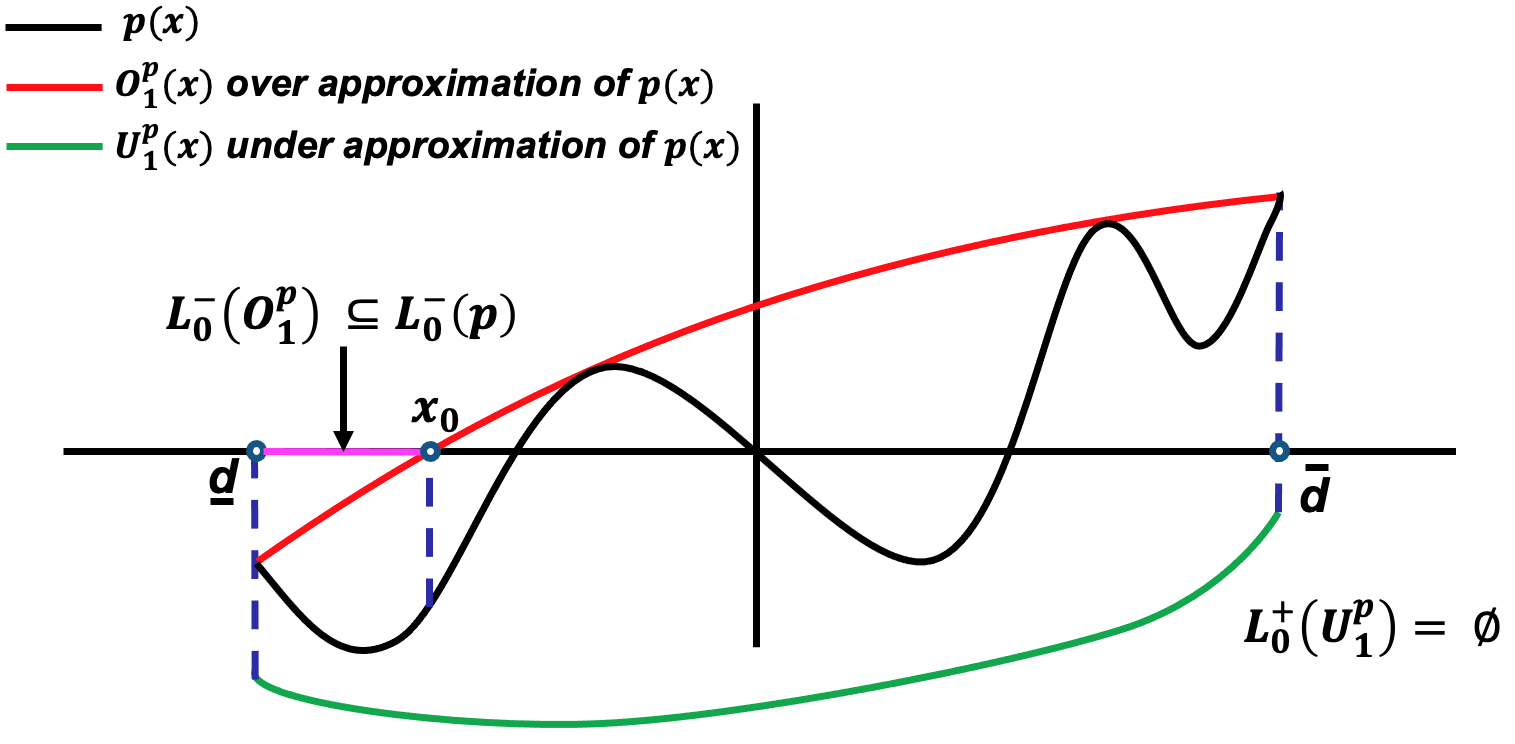} 
	\includegraphics[width=0.49\textwidth]{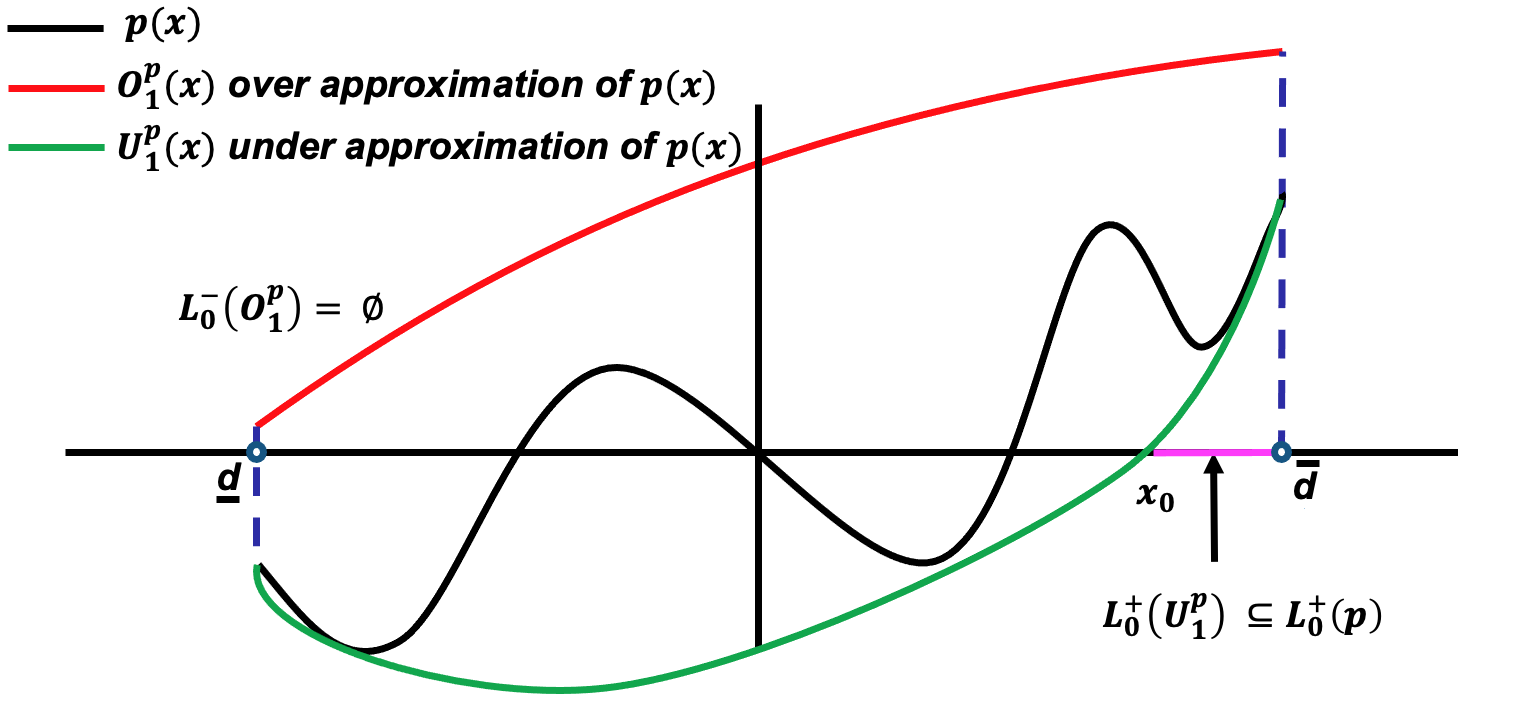} \\
	\includegraphics[width=0.49\textwidth]{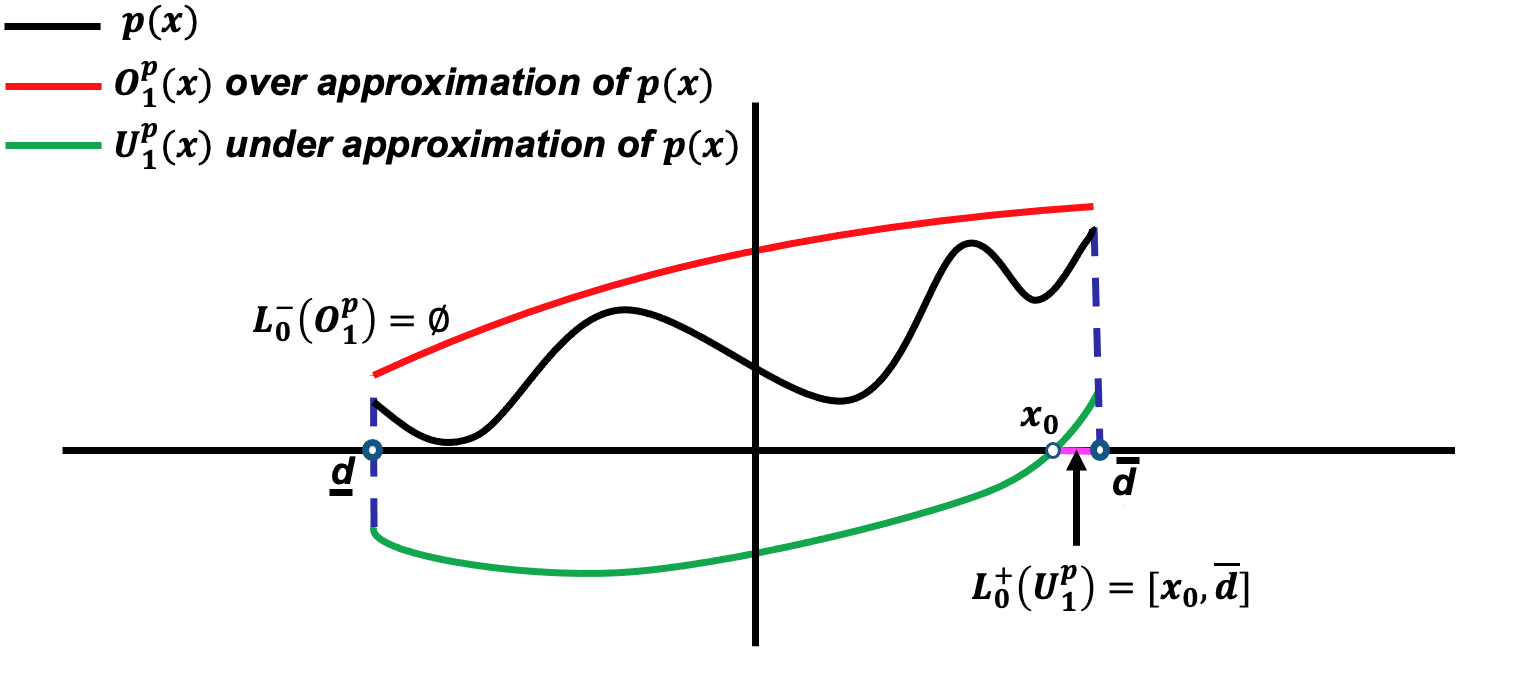} 
	\includegraphics[width=0.49\textwidth]{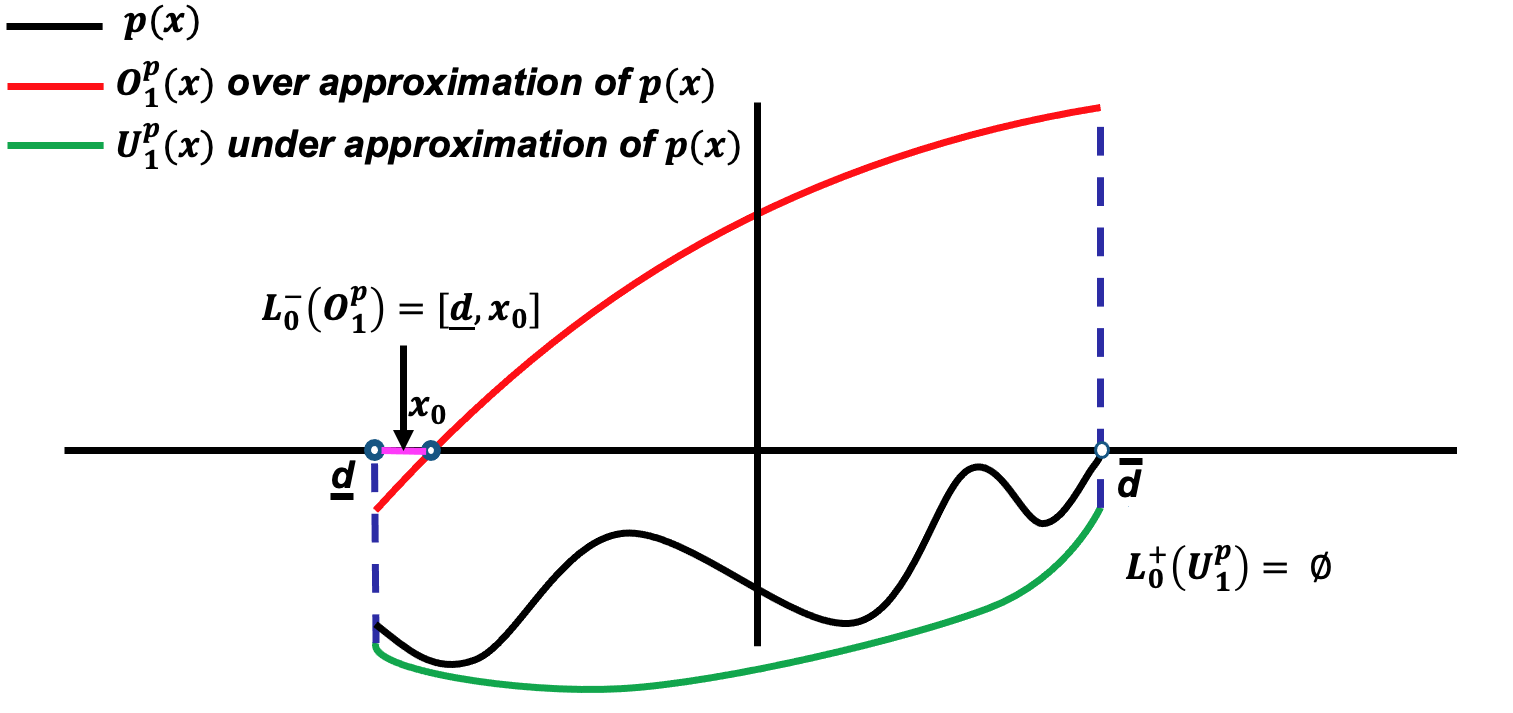}
 	\vspace{-5mm}
	\caption{Exemplary cases where abstracting higher order polynomial (black curves) using convex approximations fails to provide helpful information: \textbf{Top-Left:} under-approximation (green curve) is entirely negative and hence fails to identify any subsets of $L^{+}_{0}\left(p\right)$. \textbf{Top-Right:} over-approximation (red curve) is entirely positive and hence fails to identify subsets of $L^{-}_{0}\left(p\right)$. \textbf{Bottom:} under/over approximations failed to identify polynomials that are consistently positive (left) or negative (right).
	}
	\label{F1}
\end{figure*}

\subsection{Drawbacks of Convex Abstraction Refinement} 
Although the prescribed convex abstraction refinement process was shown to provide several orders of magnitude speedup compared to the state-of-the-art~\cite{polyar}, it adds unnecessary overhead in certain situations. In particular, and as shown in Figure~\ref{F1}, the quadratic abstractions $O^{p_i}_j(x)$ and $U^{p_i}_j(x)$ may fail to identify meaningful subsets of $L^{-}_{0}\left(p_i\right)$ and $L^{+}_{0}\left(p_i\right)$. One needs to split the input region to tighten the over-/under-approximation in such cases. Indeed, applying the convex abstraction refinement process, above, may lead to several unnecessary over-approximations or under-approximations until a tight one that prunes the search space is found. These drawbacks call for a methodology that is capable of:
\begin{enumerate}
    \item Guiding the abstraction refinement process: To reduce the number of unnecessary computations of over/under approximations, one needs a heuristic that guides the convex abstraction refinement process. In particular, such a heuristic needs to consider the properties of the polynomials and the input region to estimate the volume of the sets that the convex under/over-approximation will identify.
    
    \item Alternative Abstraction: As shown in Figure~\ref{F1} (bottom), abstracting high-order polynomials using convex ones may fail to identify easy cases when the polynomial is strictly positive or negative. Therefore, it is beneficial to use alternative ways to abstract high-order polynomials that can augment the convex abstractions.
\end{enumerate}
Designing a strategy that addresses the two requirements above is the main topic for the following two sections.

\section{Neural Network Guided Convex Abstraction Refinement}
In this section, we are interested in designing and training a Neural Network (NN) that can be used to guide the abstraction refinement process. Such NN can be used as an oracle by the solver to estimate the volume of the zero super/sub-level sets (for each polynomial) within a given region $I_n(\underline{d}, \overline{d})$ and select the best approximation strategy out of three possibilities namely: (i) apply convex under-approximation, (ii) apply convex over-approximation, and (iii) split the region to allow for finer approximations in the subsequent iterations of the solver. In this section, we aim to develop a scientific methodology that can guide the design of such NN.



\subsection{On the relation between the NN architecture and the characteristics of the polynomials:}

In this subsection, we aim to understand how the properties of the polynomials affect the design of the NN. We start by reviewing the following result from the machine learning literature:
\begin{theorem}[Theorem 1.1~\cite{NNcomplexity1}] \label{thm:errorBound}
    There exists a Rectifier Linear Unit (ReLU)-based neural network $\phi$ that can estimate a continuous function $f$ such that the estimation error is bounded by:
    $$ || \phi - f || \le \omega_f \; \sqrt{d} \; O(N^{-2/d} L^{-2/d})$$
    where $N, L, d$ are the neural network depth, the neural network width, and the number of neural network inputs, respectively, and $\omega_f$ is the Lipschitz constant of the function $f$. Moreover, this bound is nearly tight.
\end{theorem}
The above result can be interpreted as follows. The depth $N$ and width $L$ of a neural network depend on the rate of change of the underlying function (captured by its Lipschitz constant $\omega_f$). That is, if we use a NN to estimate a function with a high $\omega_f$, then one needs to increase the depth $N$ and width $L$ of the NN to achieve an acceptable estimation error.

Now we aim to connect the result above with the characteristics of the polynomials. To that end, we recall the definition of ``condition numbers'' of a polynomial~\cite{faroukinumstab}:

%
%
\begin{definition}
\label{def:condition}
Given a polynomial $p\left(x\right) = \sum\limits_{K}a_Kx^K$ and a root $x_0$ of $p$, the quantity $C_{a_p}\left(x_0\right)$ is called the condition number for the root $x_0$. The condition number characterizes the sensitivity of the root $x_0$ to a perturbation of the coefficients $a_p$. That is, if we allow a random perturbation of a fixed relative magnitude $\epsilon = \abs{\frac{\delta a_K}{a_K}}$ in each coefficient $a_K$ in $a_p$, then the magnitude of the maximum displacement $\delta x_0$ of a root $x_0$ is bounded as: $\abs{\delta x_0} \leq C_{a_p}\left(x_0\right) \epsilon$. For a polynomial with multiple roots, then we define the condition number of the polynomial $\overline{C}_{a_p}$ as the largest $C_{a_p}\left(x_0\right)$ among all roots, i.e., $\overline{C}_{a_p} = sup_{x_0 \in \{ x | p(x) = 0\}}C_{a_p}\left(x_0\right)$.
\end{definition}


We are now ready to present our first theoretical result that connects the condition number of polynomials to the NN architecture. As stated before, we are interested in designing an NN that can estimate the zero sub/super level volume set within a given region. We show that the larger the condition number, the larger the neural network depth and width, as captured by the following result.

\begin{theorem}
~\label{prop:lipNN}
Given a polynomial $p$ with coefficients $a_p$ and a region $I_n(\underline{d}, \overline{d})$.
There exists a neural network $NN(a_p, I_n)$ that estimates the volume of zero sub/super level sets from the polynomial coefficients $a_p$. The Lipschitz constant of this $NN(a_p, I_n)$, denoted by $\omega_{NN}$ is bounded by $\mathcal{O}(n_r \overline{n}_r \overline{C}_{a_p})$
where $\overline{C}_{a_p}$ is the condition number of the polynomial $p$, $ \overline{n}_r = max(n, n_r)$ and $n_r$ is the number of roots of the polynomial $p$.
\end{theorem}

To prove the result, we will proceed with an existential argument. We will show that a NN that matches the properties above exists without constructing such a NN. As shown in Figure~\ref{3NN}, the neural network $NN(a_p, I_n)$ consists of multiple sub-neural networks. In particular, the first sub-neural network $NN_{a_p \rightarrow X_0}$ computes all the roots $X_0 = (x_0^1, \ldots, x_0^{n_r})$ of the polynomial (where $n_r$ is the number of roots) from the coefficients $a_p$, i.e.:
\begin{align}
    X_0 = NN_{a_p \rightarrow X_0} (a_p).
    \label{eq:estimate}
\end{align}
Note that $NN_{a_p \rightarrow X_0}$ does not depend on the region $I_n$ and hence the roots $X_0$ may not lie inside the region $I_n$. Moreover, Theorem~\ref{prop:lipNN} asks for a NN that estimates the volume of the zero sub/super level sets and not the location of the roots. To that end, our strategy is to split the region $I_n$ into sub-regions of fixed volume and check if a root lies within each of these sub-regions. If a sub-region does not have a root (i.e., there is no zero crossing inside this sub-region) and the evaluation of the polynomial at any point in this region turns to be positive, then this sub-region belongs to the super level set of $p$ and similarly for the sublevel set of $p$. By counting the number of the sub-regions with no zero crossings and multiplying this count by the volume of these sub-regions, we can provide an estimate of the sub/super level sets. Such a process can be performed using the following three sub-neural networks:

\begin{itemize}
    \item The sub-neural network $NN_{I_n \rightarrow I^{i}_n}$ splits the region $I_n$ into $l$ sub-regions $I^{1}_n, \ldots, I^{l}_n$ and return the bounds of the $i$th sub-region, i.e.:
    \begin{align}
        (\underline{d}^i, \overline{d}^i) & = NN_{I_n \rightarrow I^{i}_n} (I_n), \qquad i \in \{1, \ldots, l\}.
        \label{eq:split}
    \end{align}
    
    \item The sub-neural network $NN_{X_0 \rightarrow ZC_i}$ checks the location of the roots $(x^1_0, \ldots, x_0^{n_r})$ and returns a binary indicator variable $ZC_i$ that indicates whether a zero-crossing takes place within the $i$th sub-region or whether the polynomial is always positive/negative within the $i$th sub-region, i.e.:
    \begin{align}
        ZC_i(a_p, I^i_n) &= NN_{X_0 \rightarrow ZC_i} \left( NN_{a_p \rightarrow X_0} (a_p), NN_{I_n \rightarrow I^{i}_n} (I_n)\right).
        \label{eq:indicator}
    \end{align}
    
    \item The final output $NN(a_p, I_n)$ is computed using the sub-neural network $NN_{ZC \rightarrow L^{+}/L^{-}}$ which counts the number of regions that has no zero-crossing (using the indicators $ZC_1, \ldots, ZC_l$) and compute the estimate of the zero sub/super level sets, i.e.:
    \begin{align}
        NN(a_p, I_n) &= NN_{ZC \rightarrow L^{+}/L^{-}} \left(ZC_1(a_p, I^1_n), \ldots, ZC_l(a_p, I^l_n) \right).
        \label{eq:count}
    \end{align}
\end{itemize}


The Lipschitz constants of these sub-neural networks are captured by the following four propositions whose proof can be found in the appendix.

\begin{proposition}
\label{prop:nn_estimate}
Consider the sub-neural network $NN_{a_p \rightarrow X_0} (a_p)$ defined in~\eqref{eq:estimate}. The Lipschitz constant of $NN_{a_p \rightarrow X_0} (a_p)$ is bounded by $\mathcal{O}(n_r \overline{C}_{a_p})$.
\end{proposition}


\begin{proposition}
\label{prop:nn_split}
Consider the sub-neural network $NN_{I_n \rightarrow I^{i}_n}$ defined in~\eqref{eq:split}. The Lipschitz constant of $NN_{I_n \rightarrow I^{i}_n}$ is bounded by $\mathcal{O}(n)$.
\end{proposition}

\begin{figure}
    \centering
    \includegraphics[width=0.5\columnwidth]{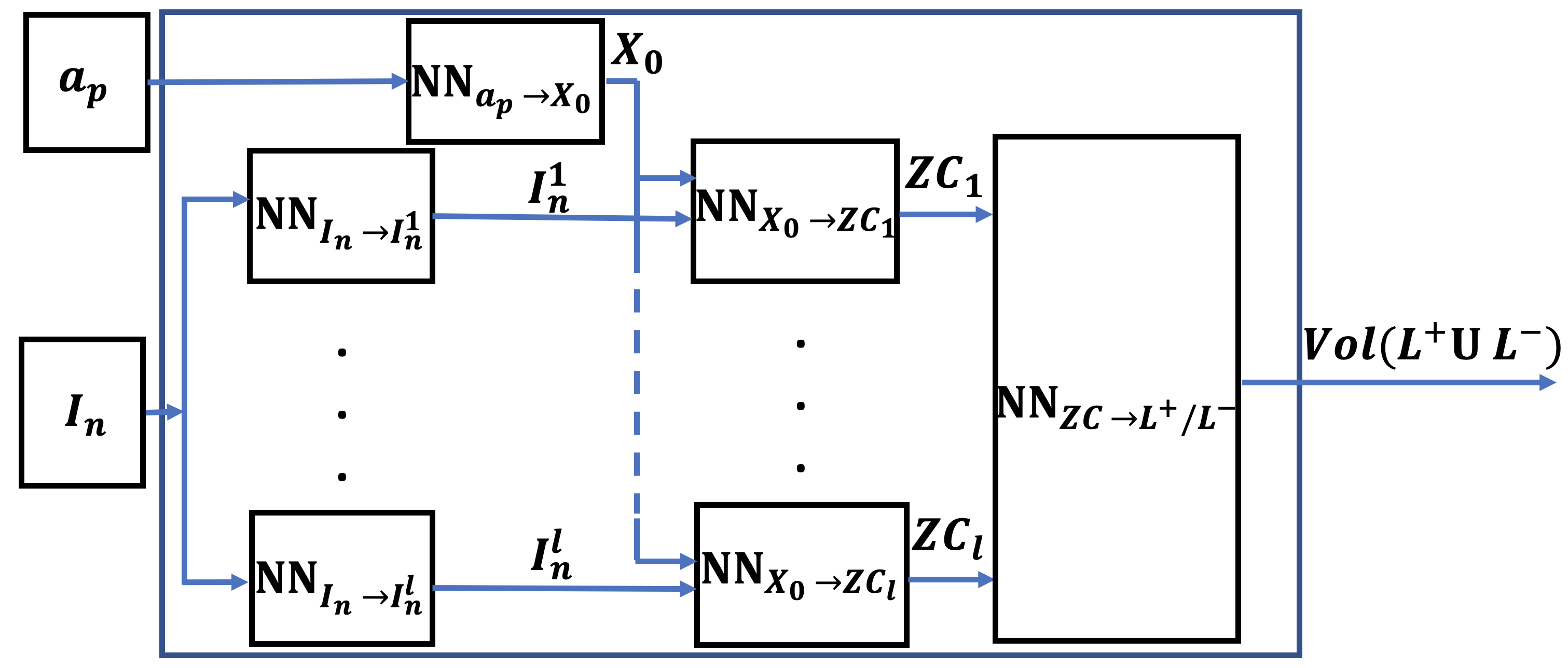}
    \caption{The architecture of the neural network $NN(a_p, I_n)$ used to prove the Theorem \ref{prop:lipNN}.}
    \label{3NN}
\end{figure}

\begin{proposition}
\label{prop:nn_indicator}
Consider the sub-neural network $NN_{X_0 \rightarrow ZC_i}\left(X_0, I^i_n\right)$ defined in~\eqref{eq:indicator}. The Lipschitz constant of $NN_{X_0 \rightarrow ZC_i}\left(X_0, I^i_n\right)$ is bounded by $\mathcal{O}(n_r)$.
\end{proposition}

\begin{proposition}
\label{prop:nn_count}
Consider the sub-neural network $NN_{ZC \rightarrow L^{+}/L^{-}}$ defined in~\eqref{eq:count}. The Lipschitz constant of $NN_{ZC \rightarrow L^{+}/L^{-}}$ is bounded by $\mathcal{O}(1)$.

\end{proposition}

\begin{proof}[Proof of Theorem~\ref{prop:lipNN}]

Consider the NN shown in Figure~\ref{3NN} and defined using equations~\eqref{eq:estimate}-\eqref{eq:count}. To bound the Lipschitz constant of $NN(a_p, I_n)$, we consider two sets of inputs $(a_p, I_n)$ and $(a'_p, I'_n)$ as follows:
\begin{align}\label{theo2_1}
\norm{NN(a'_p, I'_n) - NN(a_p, I_n)}_2 &= \left\Vert NN_{ZC \rightarrow L^{+}/L^{-}} \left(ZC_1(a'_p, I'^1_n), \ldots, ZC_l(a'_p, I'^l_n) \right) \right. \nonumber\\ 
& \left. \qquad \qquad - NN_{ZC \rightarrow L^{+}/L^{-}} \left(ZC_1(a_p, I^1_n), \ldots, ZC_l(a_p, I^l_n) \right) \right \Vert_2,  \\
& \leq \mathcal{O}(1) \norm{ \left(ZC_1(a'_p, I'^1_n) - ZC_1(a_p, I^1_n), \cdots,    ZC_l(a'_p, I'^l_n) - ZC_l(a_p, I^l_n)\right)}_2, \label{eq:apply_lip1}\\
& = \mathcal{O}(1) \bigg(\sum\limits_{i = 1}^{l} \norm{ZC_i(a'_p, I'^i_n) - ZC_i(a_p, I^i_n)}^2_2\bigg)^{\frac{1}{2}}. \label{eq:apply_lip}
\end{align}
%
where~\eqref{eq:apply_lip1} follows from Proposition~\ref{prop:nn_count}. Now, we upper bound $\norm{ZC_i(a'_p, I'^i_n) - ZC_i(a_p, I^i_n)}_2$ as follows:


\begin{align}\label{theo2_2}
 \norm{ZC_i(a'_p, I'^i_n) - ZC_i(a_p, I^i_n)}^2_2 
 &= \left\Vert NN_{X_0 \rightarrow ZC_i} \left( NN_{a_p \rightarrow X_0} (a'_p), NN_{I_n \rightarrow I^{i}_n} (I'_n)\right) \right. \nonumber \\
 &\left. \qquad \qquad \qquad  - NN_{X_0 \rightarrow ZC_i} \left( NN_{a_p \rightarrow X_0} (a_p), NN_{I_n \rightarrow I^{i}_n} (I_n)\right)\right \Vert^2_2 \\ 
 & \leq \mathcal{O}(n_r) \norm{ \left( NN_{a_p \rightarrow X_0} (a'_p) - NN_{a_p \rightarrow X_0} (a_p), NN_{I_n \rightarrow I^{i}_n} (I'_n) - NN_{I_n \rightarrow I^{i}_n} (I_n)    \right)  }^2_2, \label{eq:apply_lip2}\\
 & = \mathcal{O}(n_r) \left (\norm{ \left( NN_{a_p \rightarrow X_0} (a'_p) - NN_{a_p \rightarrow X_0} (a_p) \right) }^2_2 +  \norm{NN_{I_n \rightarrow I^{i}_n} (I'_n) - NN_{I_n \rightarrow I^{i}_n} (I_n)  }^2_2 \right), \nonumber \\
 & \leq \mathcal{O}(n_r) \left(\mathcal{O}(n_r \overline{C}_{a_p}) \norm{a'_p - a_p}^2_2 + \mathcal{O}(n) \norm{I'_n - I_n}^2_2 \right) \label{eq:apply_lip3} \\
 & = \mathcal{O}(n_r \overline{n}_r \overline{C}_{a_p}) \norm{(a'_p,I'_n) - (a_p,I_n)}_2^2. \label{eq:apply_lip4}
\end{align}
where~\eqref{eq:apply_lip2} follows from Proposition~\ref{prop:nn_indicator};~\eqref{eq:apply_lip3} follows from Propositions~\ref{prop:nn_split} and~\ref{prop:nn_estimate} along with the definition of $\overline{n}_r = \max(n, n_r)$. Substituting~\eqref{eq:apply_lip4} in~\eqref{eq:apply_lip} and noticing that $l$ is a constant that does not depend on $n$ yields:
%
\begin{align}\label{theo2_3}
  \norm{NN(a'_p, I'_n) - NN(a_p, I_n)}_2 \leq \mathcal{O}(n_r \overline{n}_r \overline{C}_{a_p}) \norm{(a'_p,I'_n) - (a_p,I_n)}_2,  
\end{align}
from which we conclude that the Lipschitz constant of $NN(a_p, I_n)$ is in the order of $\mathcal{O}(n_r \overline{n}_r \overline{C}_{a_p})$ which concludes the proof of Theorem~\ref{prop:lipNN}.

\end{proof}

It follows from Theorem~\ref{thm:errorBound} and Theorem~\ref{prop:lipNN} that the higher the condition number of a polynomial $\overline{C}_{a_p}$, the larger the network width and depth needed to estimate the volume of zero sub/super level sets with high accuracy. Unfortunately, the power basis representation of polynomials (i.e., representing the polynomial as a summation $\sum_{K \le L}a_K x^K$) is shown to be an unstable representation with extremely large condition numbers~\cite{faroukinumstab} which may necessitate neural networks with substantial architecture.

\subsection{Bernstein Polynomials: A Robust Representation of Polynomials}

Motivated by the challenge above, we seek a representation of polynomials that is more robust to changes in coefficients, i.e., we seek a representation in which the roots of the polynomial change slowly with changes in the coefficients (and hence smaller condition numbers $\overline{C}_{a_p}$ and a smaller NN to estimate the volume of the sub/super level sets). We start with the following definition.

\begin{definition}
Let $p\left(x\right) = \sum\limits_{K \leq L}^{} a_K x^K \in \mathbb{R}[x_1,\ldots,x_n]$ be a multivariate polynomial over a hyperrectangle $I_n (\underline{d}, \overline{d})$
and of a maximal degree $L = \left(l_1, \cdots, l_n\right) \in \mathbb{N}^n$. 
The polynomial: 
\begin{align}\label{bernpol}
B_{p, L}\left(x\right) &= \sum\limits_{K \leq L}^{} b_{K,L} Ber_{K, L}\left(x\right),
\end{align}
is called the Bernstein polynomial of $p$, where $Ber_{K, L}\left(x\right)$ and $b_{K, L}$ are called the Bernstein basis and Bernstein coefficients of $p$, respectively, and are defined as follows: 
\begin{align}\label{bernpolcoeff}
Ber_{K, L}\left(x\right) &= {L \choose K} x^{K}\left(1-x\right)^{L-K}, \qquad
b_{K, L} =\sum\limits_{J = (0,\ldots, 0)}^{K} \frac{{K \choose J}}{{L\choose J}}\left(\overline{d} - \underline{d}\right)^J \sum\limits_{I = J}^{L}{I \choose J} \underline{d}^{I - J} a_I.
\end{align}
\end{definition}


The Bernstein representation is known to be the most robust representation of polynomials which is captured by the next result~\cite{faroukinumstab}.
\begin{theorem}[Theorem~\cite{faroukinumstab}] \label{thm:bernStable}
The Bernstein basis is optimally stable, i.e. there exists no other basis with a condition number smaller than the condition number of the Bernstein coefficients $\overline{C}_{b_p}$, where $b_p = (b_{(0,0,\ldots,0), L}, \ldots, b_{(l_1, l_2, \ldots, l_n), L})$ is the vector of all the Bernstein coefficients of polynomial $p$. 
\end{theorem}


Theorems~\ref{thm:errorBound}-\ref{thm:bernStable} point to the optimal way of designing the targeted neural network. Such a neural network needs to take as input the Bernstein coefficients $b_{p}$ instead of the power basis coefficients $a_p$. To validate this conclusion, we report empirical evidence in Table~\ref{TabComBases}. In this numerical experiment, we trained two neural networks with the same exact architecture, using the same exact number of data points, and both networks have the same number of inputs. Both neural networks are trained to estimate whether a zero-crossing occurs in a region (recall from our analysis in Theorem~\ref{prop:lipNN} that the Lipschitz constant of this NN is equal to the condition number of the polynomial). The only difference is that one neural network is trained using power basis coefficients $a_p$ (column 3 of Table~\ref{TabComBases}) while the second is trained using Bernstein basis coefficients $b_{p}$ (column 5 of Table~\ref{TabComBases}). The coefficients are randomly generated via a uniform distribution between $-0.1$ and $0.1$, i.e., $\mathcal{U}\left(-0.1, 0.1\right)$. We generated 40000 training samples and 10000 validation samples for both bases. We evaluate the trained NN on three different benchmarks for the two bases. Each evaluation benchmark has 10000 samples. The results are summarized in Table~\ref{TabComBases}. As it can be seen from Table~\ref{TabComBases}, the NN trained with Bernstein coefficients generalizes better than the NN trained with power basis coefficients as reflected by the empirical ``Accuracy'' during evaluation. This empirical evidence matches our analysis in Theorem~\ref{prop:lipNN} along with the insights of Theorem~\ref{thm:errorBound} and Theorem~\ref{thm:bernStable}.


\begin{table}[t!]
\caption{Evaluation of three trained neural networks on three different benchmarks for the different polynomial basis. Each benchmark has 10000 samples. The coefficients of the polynomial within each basis are generated following a uniform distribution given in the table.} \vspace{-3mm}
\begin{adjustbox}{width=0.8\columnwidth,center}
\begin{tabular}{|c|c|c|c|c|c|c|c|}
    \hline
      Benchmark & Coefficients  &   \multicolumn{2}{c|}{Power Basis} & \multicolumn{2}{c|}{Bernstein Basis} & \multicolumn{2}{c|}{Reduced Bernstein Basis}  \\ 
     \cline{3-8}
     & & Accuracy & Overhead & Accuracy & Overhead & Accuracy & Overhead\\
    \hline
   1  & $\mathcal{U}\left(-0.1, 0.1\right)$ &  $46\%$ & 0 [s] & $91\%$  & 0.01 [s] & $82\%$ & 0.002 [s]\\
    \hline 
   2 & $\mathcal{U}\left(-0.5, 0.5\right)$  &  $32\%$ &  0 [s] & $87\%$& 0.03 [s] & $79\%$ & 0.005 [s]\\
     \hline 
   3  & $\mathcal{U}\left(-1, 1\right)$  &  $30\%$ & 0 [s] & $88\%$  & 0.04 [s] & $80\%$ & 0.007 [s]\\
      \hline
\end{tabular}
\end{adjustbox}
\label{TabComBases}
\end{table}

\section{Taming the Complexity of Computing Bernstein Coefficients}

In section 4, we concluded that Bernstein's representation has a smaller condition compared to other representations, which helps build a more efficient NN. Nevertheless, computing this representation adds a significant overhead even by using the most efficient algorithms to calculate these  coefficients~\cite{range4,berncomplex}.
%
For example, computing all the Bernstein coefficients of a $6^{th}$-dimensional polynomial with $7^{th}$ order using  Matrix method and Garloff’s methods~\cite{range4,berncomplex} require $1.1e07$ and $7.1e06$ summation and multiplication operations~\cite{berncomplex}.
To exacerbate the problem, the Bernstein coefficients depend on the region $I_n$ and need to be recomputed in every iteration of the abstraction refinement process. Reducing such overhead is the main focus of this section.

\subsection{Range Enclosure Property of Bernstein polynomials}

Given a multivariate polynomial $p\left(x\right)$ that is defined over the $n$-dimensional box $I_n(\underline{d}, \overline{d})$, we can bound the range of $p\left(x\right)$ over $I_n(\underline{d}, \overline{d})$ using the range enclosure property of Bernstein polynomials as follows:
\begin{theorem}[Theorem 2 
\cite{garloff}]\label{th1}
Let $p$ be a multivariate polynomial of degree $L$ over the $n$-dimensional box $I_n(\underline{d}, \overline{d})$ with Bernstein coefficients $b_{K, L}$, $0 \leq K \leq L$. Then, for all $x \in I_n$, the following inequality holds:
\begin{align}\label{bernbound}
    \min_{K \leq L}~b_{K, L} \leq p\left(x\right) \leq \max_{K \leq L}~b_{K, L}.
\end{align}
\end{theorem}


The traditional approach to computing the range enclosure of $p$ is to compute all the Bernstein coefficients of $p$ to determine their minimum and maximum~\cite{range1,range2,range3}. However, computing all the coefficients has a complexity of $\mathcal{O}\left(\left(l_{max} + 1\right)^n\right)$, where $l_{max} = \max\limits_{1 \leq i \leq n}l_i$, which increases exponentially with the dimension $n$. 
Luckily, the Bernstein coefficients enjoy monotonicity properties, whenever the region $I_n(\underline{d}, \overline{d})$ is restricted to be an orthant (i.e., the sign of $x_i$ does not change within $I_n(\underline{d}, \overline{d})$, for each $i \in \{1,\ldots, n\}$)~\cite{range4}. Using such monotonicity properties, one can compute the minimum and maximum Bernstein coefficients (denoted by $\underline{B}_{p, L}$, $\overline{B}_{p, L}$) with a time complexity of $\mathcal{O}\left(2\left(l_{max} + 1\right)^2\right)$ which does not depend on the dimension $n$.

 \subsection{Zero Crossing Estimation using only a few Bernstein Coefficients}
 Now we discuss how to use the range enclosure property above to reduce the number of computed Bernstein coefficients. First, we note that the zero crossing of a polynomial $p$ in a given input region $I_n$ depends on its estimate range given by $\underline{B}_{p, L}$ and $\overline{B}_{p, L}$. More specifically, if $\underline{B}_{p, L} > 0$ ($\overline{B}_{p, L} < 0$),  then the entire polynomial is positive (negative), which means that there is no zero-crossing. If $\underline{B}_{p, L}$ and $\overline{B}_{p, L}$ have different signs, and because of the estimation error of these bounds, the polynomial $p$ may still be positive, negative, or have a zero crossing in the region. In this case, we need additional information such as the bounds of the gradient of the polynomial $p$ within the input region, that are given by $\underline{B}_{\nabla p, L}$ and $\overline{B}_{\nabla p, L}$ (which can be computed efficiently thanks to the fact that gradients of polynomials are polynomials themselves). Such additional information about the worst-case gradient of the polynomial leads to a natural estimate of whether a zero crossing occurs in a region. 
 
 
Due to space constraints, we omit the analysis of bounding the estimation error introduced by relying only on the maximum and minimum of the polynomial $\underline{B}_{p, L}$ and $\overline{B}_{p, L}$ along with the maximum and minimum of the gradient $\underline{B}_{\nabla p, L}$ and $\overline{B}_{\nabla p, L}$. Instead, we support our claim using the empirical evidence shown in Table~\ref{TabComBases}. Using the same benchmarks used in Section 4.2, we train a third neural network that takes as input only the four inputs $\underline{B}_{p, L}, \overline{B}_{p, L}, \underline{B}_{\nabla p, L}, \overline{B}_{\nabla p, L}$ and compare its generalization performance (column 7 of Table~\ref{TabComBases}). As shown in the table, the third neural network sacrifices some accuracy compared to the ones that use all Bernstein coefficients. But on the other side, it reduces the overhead to compute the Bernstein coefficients by order of magnitude as can be seen by comparing the ``execution overhead'' reported in columns 4, 6, and 8 for the power basis, the Bernstein basis, and the reduced Bernstein basis, respectively.

\subsection{Search Space pruning using Bernstein Coefficients}
The range enclosure property and the discussion above open the door for a natural solution of the ``alternative abstraction'' problem mentioned in Section 3.2. The maximum and minimum Bernstein coefficients can be used as an abstraction (in addition to convex upper and lower bounds) of high-order polynomials. Such abstractions can be refined with every iteration of the solver. They can be used to identify portions of the search space for which one of the polynomials is guaranteed to be positive (and hence a solution does not exist). More details about integrating this abstraction and the convex abstraction are given in the implementation section below.

\section{Algorithm Architecture and Implementation Details}


\begin{algorithm}[t!]
\caption{PolyARBerNN}  \label{alg:PolyARBerNN}
\begin{flushleft}
\textbf{Input:} $I_n(\underline{d}, \overline{d}), p_1, p_2, \ldots, p_m, \epsilon$, $\qquad$
\textbf{Output}: $x_{\text{Sol}}$
\end{flushleft}
\begin{algorithmic}[1]
\STATE $orthants := \texttt{Partition\_Region}(I_n)$
\STATE $Neg := \{\}$
\STATE $Ambig := \{orthants\}$
\STATE $\text{List\_pols} := \{p_1,\ldots,p_m\}$
\WHILE{$\texttt{Compute\_Maximum\_Volume}(Ambig) \ge \epsilon$}
    \STATE $p :=\texttt{Select\_Poly}\left(\text{List\_pols}, Neg\right)$
    \STATE $region := \texttt{Remove\_Ambiguous\_Region\_From\_List}\left(Ambig\right)$
    \STATE $\big(\underline{B}_{p, L} ,\overline{B}_{p, L}, \underline{B}_{\nabla p, L}, \overline{B}_{\nabla p, L}\big)$ := $\texttt{Compute\_Bern\_Coeff}(p, region)$\\
    \IF{$\underline{B}_{p, L} > 0$} 
        \STATE \textbf{break}
    \ELSIF{$\underline{B}_{p, L} < 0$}
        \STATE $Neg := Neg \cup (p, L^{-}_{0}\left(p\right))$
        \STATE \textbf{break}
    \ENDIF
    \STATE $\text{(under\_approx, over\_approx, split)} := NN(\underline{B}_{p, L} ,\overline{B}_{p, L}, \underline{B}_{\nabla p, L}, \overline{B}_{\nabla p, L}, region\big)$
    \STATE $action := \texttt{Select\_best\_action}(\text{under\_approx, over\_approx, split})$
    \STATE $L^{-}_{0}\left(p\right),L^{+}_{0}\left(p\right),L^{+/-}_{0}\left(p\right):= \texttt{Convex\_Abst\_Refin\_PolyAR}\left(p, action, region\right)$
    \STATE $Ambig := Amibg \cup L^{+/-}_{0}\left(p\right)$
    \STATE $Neg := Neg \cup (p, L^{-}_{0}\left(p\right))$
\ENDWHILE
\IF{$\texttt{is\_List\_Empty}(Ambig)$}
    \IF{A negative region in $Neg$ has all the polynomials}
        \STATE $x_{\text{Sol}}:=$ any point in the negative region
    \ELSE
    \RETURN $\text{the problem is UNSAT}$
    \ENDIF
\ELSE
\STATE $x_{\text{Sol}}:=\texttt{CAD\_Solver\_Parallel}\left(Ambig,p_1, \ldots, p_m\right)$
\ENDIF

  

    


     
         
\end{algorithmic}
\end{algorithm}

In this section, we describe the implementation details of our solver PolyARBerNN. As a pre-processing step, the tool divides the input region $I_n$ into several regions such that each one is an orthant. This allows the tool to process each orthant in parallel or sequentially. The tool keeps track of all regions for which the sign of a polynomial is not fixed. These regions are called ambiguous regions, and they are stored in a list called $Ambig$. As long as the volume of the regions in this list is larger than a user-defined threshold $\epsilon$, then our tool will continuously use abstractions to identify portions in which one of the polynomials is always positive (and hence removed from the search space) or negative (and hence the tool will give higher priority for this region). The abstraction refinement is iteratively applied in Lines 5-17 of Algorithm 1. In each abstraction refinement step, the tool picks a polynomial $p$ and a region $region$ based on several heuristics (Lines 6-7). In lines 8-14, we compute the maximum/minimum Bernstein coefficients followed by checking the sign of the polynomial within this region. Suppose the Bernstein coefficients indicate that the polynomial is always positive in this region. In that case, this provides a guarantee that a solution does not exist in this region (recall that Problem 1 searches for a point where \emph{all} polynomials are negative). Similarly, if the polynomial is always negative, then it will be added to the list of negative regions. For those polynomials for which the Bernstein abstraction failed to identify their signs, we query the trained neural network to estimate the best convex abstraction possible (Lines 15-16). Based on the neural network suggestion, we use the PolyAR tool~\cite{polyar} to compute the convex abstraction (Line 17), which returns portions of this region that are guaranteed to belong to the zero sublevel set $L^-_0(p)$, those who belong to the zero superlevel set $L^+_0(p)$, and those remain ambiguous $L^{+/-}_0(p)$. The process of using Bernstein abstraction and the convex abstraction (which is guided by the trained neural network) continues until all remaining ambiguous regions are smaller than a user-defined threshold $\epsilon$ in which case it will be processed in parallel using a sound and complete tool that implements Cylindrical Algebraic Decomposition (CAD) such as Z3 and Yices (line 28 in Algorithm 1).

\begin{figure}[!t]
	\includegraphics[width=0.7\columnwidth]{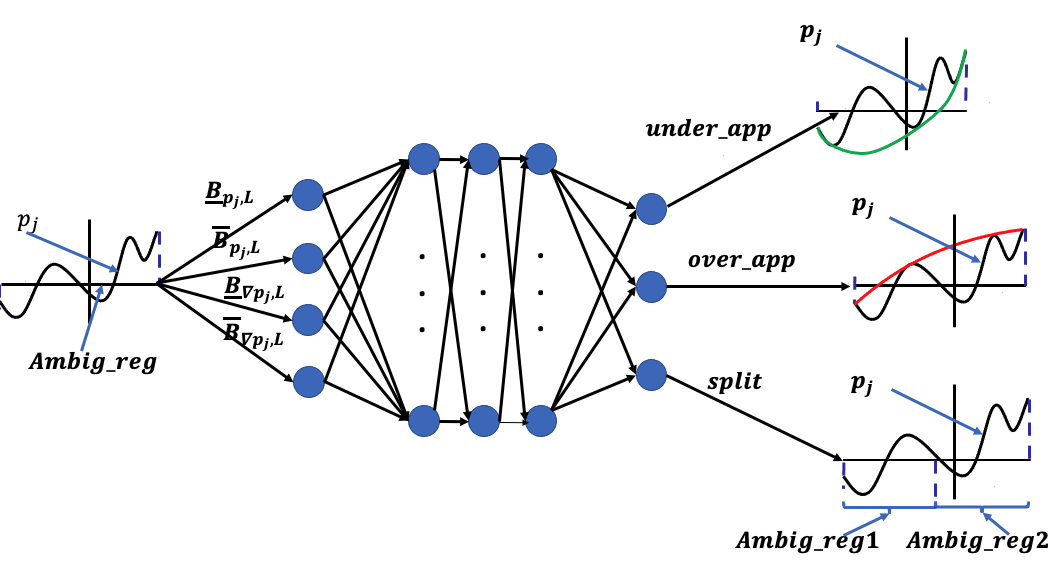} 
	\vspace{-3mm}
	\caption{The architecture of the trained NN that is used to guide the abstraction refinement process within PolyARBerNN.
	We used a fully connected NN that contains an input layer with 4 neurons, three hidden layers with 40 neurons each, and one output layer with three neurons. All neurons are ReLU-based except for the output neurons which uses SoftMax non-linearity.}
	\label{fig:AbstrefinNN}
	\vspace{-5mm}
\end{figure}

The neural network itself is trained using randomly generated, quadratic, two-dimensional polynomials where the coefficients follow a uniform distribution between $-1$ and $1$. For each randomly generated polynomial, we used PolyAR to compute the volumes of the $L^+_0(p), L^-_0(p), L^{+/-}_0(p)$ regions. We use a fully connected NN that contains an input layer, three hidden layers, and one output layer (shown in Figure~\ref{fig:AbstrefinNN}). The input layer has four neurons, the hidden layers have 40 neurons each, and the output layer has three neurons. We use a dropout of probability $0.5$ in the first and second hidden layers to avoid overfitting. We use the ReLU activation function for all the hidden layers and the Softmax activation function for the output layer. We use Adam as an optimizer and cross-entropy as a loss function. Although the neural network is trained on simple quadratic two-dimensional polynomials, we observed it generalizes well to higher-order polynomials with several variables. This will become apparent during the numerical evaluation in which polynomials of different orders and several variables will be used to evaluate the tool.

~\\
\noindent\textbf{Correctness Guarantees:}
We conclude our discussion with the following result which captures the correctness guarantees of the proposed tool:
\begin{theorem}
The PolyARBerNN solver is sound and complete.
\end{theorem}
\begin{proof}
This result follows from the fact that search space is pruned using sound abstractions (convex upper bounds or Bernstein-based). The neural network and the convex lower bound polynomials are just used as heuristics to guide the refinement process. Finally, CAD-based algorithms (which are sound and complete) are used to process the portions of the search space which are not pruned by the abstraction refinement.
\end{proof}

\section{Generalization to polynomial optimization problems:}



In this section, we focus on providing a solution to Problem 2. Our approach is to turn the optimization problem (Problem 2) into a feasibility problem (Problem 1). First, we recall that the gradient of $p$, $\nabla p = [\frac{\partial p}{\partial x_1}, \cdots{}, \frac{\partial p}{\partial x_n} ]$, where $\frac{\partial p}{\partial x_i}$ is the partial derivative of $p$ with respect to $x_i$, is a vector of $n$ polynomials. The optimal value of $p$ occurs either (i) when the vector of partial derivatives are all equal to zero or (ii) at the boundaries of the input region.


To find the critical points $x^{*}$ of $p$ where $\nabla p \left(x^{*}\right) = 0$, we add the $n$ polynomial constraints $\frac{\partial p}{\partial x_i} \le 0, 1\leq i \leq n $ and $- \frac{\partial p}{\partial x_i} \le 0, 1\leq i \leq n $ to the constraints of the optimization problem. Now, we modify the PolyARBernNN solver to output  \emph{all} possible regions in which all the constraints are satisfied. This can be easily computed by taking the intersections within the regions stored in the data structure $Neg$ in Algorithm 1. These regions enjoy the property that \emph{all} points in these regions are critical points of $p$. In addition, we modify PolyARBernNN to output \emph{all} the remaining ambiguous regions whose volumes are smaller than the user-specified threshold $\epsilon$ and for which the CAD-based solvers returned a solution. These regions enjoy the property that \emph{there exists} a point inside these regions which is a critical point. These modifications are captured in (Line 2 of Algorithm 2).



Since the minimum/maximum of $p$ may occur at the boundaries of the region $I_n\left(\underline{d}, \overline{d}\right)$, our solver samples from the boundaries of the region $I_n\left(\underline{d}, \overline{d}\right)$ (Line 4 in Algorithm 2). The solver uses $2 \sqrt{n} (\epsilon)^{1/n}$ as sampling distance between two successive boundary samples---recall $\epsilon$ is a user-defined parameter and was used in Algorithm 1 as a threshold on the refinement process. Finally, we evaluate the polynomial $p$ in the obtained samples and we take the minimum and maximum over the obtained values (Line 7 in Algorithm 2). All the details can be found in Algorithm 2. 

\begin{algorithm}[t!]
\caption{PolyAROpt}  \label{alg:PolyAROpt}
\begin{flushleft}
\textbf{Input:} $I_n(\underline{d}, \overline{d}), p, p_1, p_2, \ldots, p_m, \epsilon$\\
\textbf{Output}: $\hat{p}_{min}, \hat{p}_{max}$
\end{flushleft}
\begin{algorithmic}[1]
\STATE $\nabla p = \texttt{Grad\_Poly}(p)$
\STATE $\hat{reg}_{\text{list}} = \texttt{PolyARBerNN}(I_n(\underline{d}, \overline{d}), \nabla p, p_1, \ldots, p_m, \epsilon)$
\STATE $\hat{x}_{\text{list}} = \texttt{center}(\hat{reg}_{\text{list}})$
\STATE $x^{end}_{list} = \texttt{Sample\_boundaries}(I_n(\underline{d}, \overline{d}), \epsilon)$
\STATE $\hat{p}_{list} = p(\hat{x}_{\text{list}})$; $\quad p^{end}_{list} = p(x^{end}_{\text{list}})$
\STATE $p_{list} = \hat{p}_{list} \cup p^{end}_{list}$
\STATE $\hat{x}_{min} = \arg \min (p_{list}$); $\quad \hat{x}_{max} = \arg \max (p_{list})$
\STATE $\hat{p}_{min} = p(\hat{x}_{min})$; $\quad \hat{p}_{max} =  p(\hat{x}_{max})$
\end{algorithmic}
\end{algorithm}

We conclude our discussion with the following result which captures the error between the solutions provided by PolyAROpt and the global optima.

\begin{theorem}
Let $p^{*}_{min}$ and $p^{*}_{max}$ be the global optimal points for the solution of Problem 2. The solution obtained by Algorithm 2, denoted by $\hat{p}_{min}$ and $\hat{p}_{max}$ satisfies the following:
\begin{align}
&\norm{\hat{p}_{min} - p^{*}_{min}}\leq 2 \omega_p\sqrt{n}(\epsilon)^{1/n},   
& \norm{\hat{p}_{max} - p^{*}_{max}}\leq 2 \omega_p\sqrt{n}(\epsilon)^{1/n} \label{eq:max_bound}
\end{align}
where $\omega_p$ is the Lipschitz constant of the polynomial $p$ and $\epsilon > 0$ is a user-defined error.
\end{theorem}

\begin{proof}
We note that there are three cases that Algorithm 2 uses to compute the critical points, $\hat{x}_{min}$ and $\hat{x}_{max}$, which corresponds to the approximation of the minimum $\hat{p}_{min}$ and maximum $\hat{p}_{max}$ of the polynomial:
\begin{enumerate}
    \item Using the center of the regions in the $Neg$ list
    \item Using the center of the regions in the $Ambig$ list
    \item Using samples from the boundaries
\end{enumerate}
We proceed by case analysis. \textbf{Case 1:} First, we note that \emph{all} the points within the $Neg$ regions satisfy that $\nabla p = 0$ and hence the value of the polynomial takes the same exact value overall the region, hence the value of $p$ at the center $\hat{x}$ of the region is the same at the global optima $x^*$. \textbf{Case 2 and Case 3:}
If we can show that $\hat{x}_{min}$ and $\hat{x}_{max}$ are bounded from the actual optimal points $x^*_{min}$ and $x^*_{max}$ by:
\begin{align}
&\norm{\hat{x}_{min} - x^{*}_{min}}\leq 2 \sqrt{n}(\epsilon)^{1/n}  
& \norm{\hat{x}_{max} - x^{*}_{max}}\leq 2 \sqrt{n}(\epsilon)^{1/n} \label{eq:x_max_bound}
\end{align}
then~\eqref{eq:max_bound} will follow directly from the definition of the Lipschitz continuity of polynomials. To show that inequalities~\eqref{eq:x_max_bound} hold we proceed by case analysis. However~\eqref{eq:x_max_bound} follow directly in Case 2 from the fact that $Ambig$ regions have a volume that is smaller than $\epsilon$ (Line 7 in Algorithm 1) and hence the distance between any two points within the regions is bounded by $2 \sqrt{n}(\epsilon)^{1/n}$. Similarly, Case 3 follows from the fact that Algorithm 2 samples from the boundaries with a maximum distance between the samples that is equal to $2 \sqrt{n}(\epsilon)^{1/n}$.

\end{proof}

\section{Numerical Results - NN Training }
In this section, we show the details of training and evaluating the NN used to help PolyARBerNN selecting the best convex abstraction. We evaluate the trained NN on six different benchmarks. The benchmarks are different than the training benchmarks with respect to the input region, the degree of the polynomial, and the number of variables of the polynomial. 
All the experiments were executed on an Intel Core i7 2.6-GHz processor with 16 GB of memory.

\subsection{Training data collection and pre-processing}

\subsubsection{Data collection}
To collect the data, we generated random quadratic two-dimensional polynomials: 
$$q\left(x_1,~x_2\right)~=~c_1x_1^2 + c_2x_2 + c_3x_1x_2 + c_4x_2 + c_5y + c_6,$$
where the coefficients $c_1,\ldots, c_6$ follow a uniform distribution between $-1$ and $1$.
The random generated polynomials are defined over the domain $I_2 = \big[-2, 2\big]^2$. For each randomly generated polynomial, we perform the abstraction refinement on the domain $I_2$, iteratively. In every iteration, we perform under-approximation, over-approximation of the original polynomial over a selected ambiguous region, and a split of the ambiguous region. Next, we compute the volume of the remaining ambiguous region after each action was implemented. 
%
%
The labels are a one-hot vector of dimension three where each component represents the action that leads to the maximum reduction in the volume of the ambiguous region, either under-approximation, over-approximation or divide the region into two regions.
We ran the abstraction refinement process on all the generated polynomials to collect the data $\big(\underline{B^i}_{p_j, L} ,\overline{B^i}_{p_j, L}, \underline{B^i}_{\nabla p_j, L}, \overline{B^i}_{\nabla p_j, L}\big)$, where $i$ denote the index of the sample. We generate $50000$ samples for training, $10000$ samples for validation, and $10000$ for testing.

\subsubsection{Data Normalization}
In the literature of NN \cite{lecun}, it is important to normalize the data when the data vary across a wide range of values. This normalization leads to faster training and improves the generalization performance of the NN \cite{lecun}. Therefore, we normalize all the input data to a zero mean and unit variance by adopting a simple affine transformation $\text{data\_sample} \leftarrow \frac{\text{data\_sample} - \mu}{\sigma}$, where $\mu$ and $\sigma$ are the mean of the data and its standard deviation. The $\mu$ and $\sigma$ parameters are initialized with respectively the empirical mean and standard deviation of the dataset and they are computed offline before the training.





\begin{figure*}[!t]
    \centering
	\includegraphics[width=0.32\textwidth]{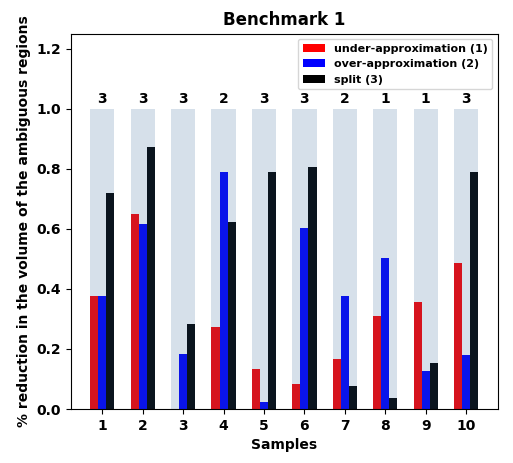}
	\includegraphics[width=0.32\textwidth]{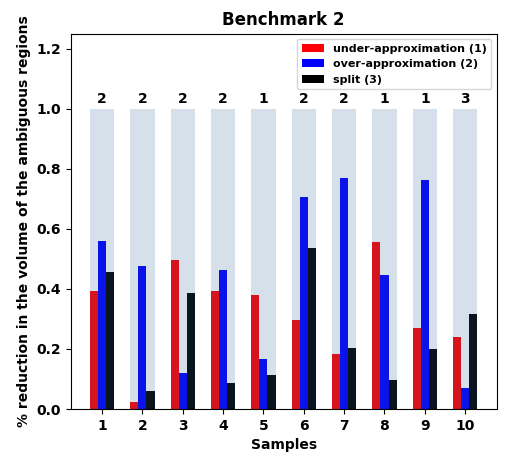}
	\includegraphics[width=0.32\textwidth]{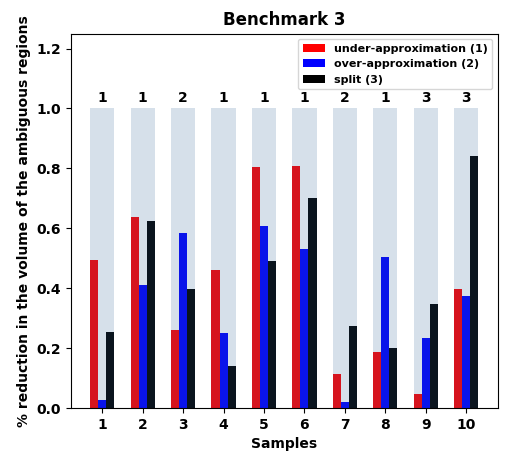}
	\includegraphics[width=0.32\textwidth]{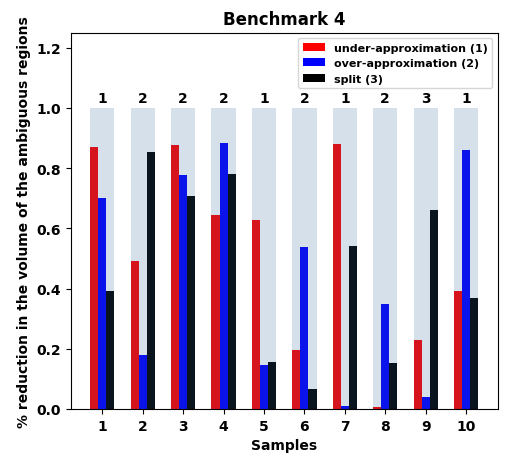}
	\includegraphics[width=0.32\textwidth]{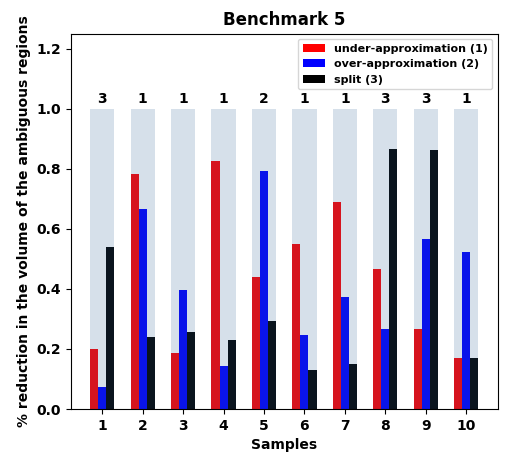}
	\includegraphics[width=0.32\textwidth]{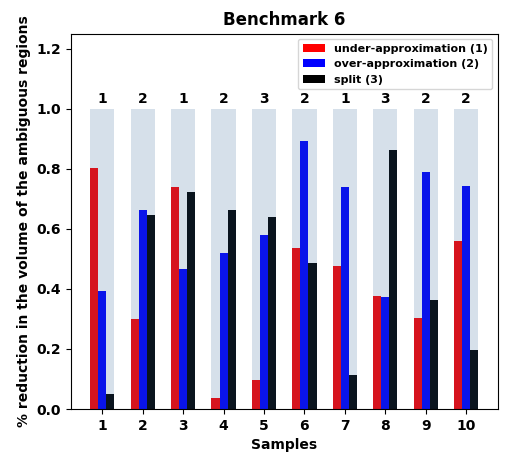}
 	\vspace{-3mm}
	\caption{Percentage in reduction of the volume of ambiguous regions along with the NN output number (the number is at the top of histograms) for 20 samples for 6 evaluation benchmarks described in Table II.}
	\label{fig:nn_performance}
	\vspace{-5mm}
\end{figure*}

 \subsection{NN's evaluation}
 We evaluated the NN on 6 different benchmarks as follows:
 \begin{itemize}
     \item  In the first and second benchmarks, we generate the same random quadratic polynomials but in a different domain: $\big[-4, 4\big]^2$ and $\big[-10, 10\big]^2$. This choice is made to test the generalization of the NN outside the data domains that were used in its training. This is important since the Bernstein coefficients of a polynomial (the input to the NN) depend on the input region $I_n$.
     
     \item In the third and fourth benchmarks, we generated random polynomials with degrees $4$ and $10$ over the domain $\big[-2, 2\big]^2$. These benchmarks are used to validate the generalization of the NN to polynomials of orders higher than the ones used in its training.
     
     \item Finally, in the fifth and sixth benchmarks, we generated random polynomials with higher dimensions, i.e., with dimension $n = 4$ and $n = 7$. 
 \end{itemize}
In summary, these benchmarks will help us to answer the following question: can the trained NN generalize to new data with different domains (benchmarks 1 and 2), higher orders (benchmarks 3 and 4), and higher dimensions (benchmarks 5 and 6)? More detail about the different benchmarks is shown in Table \ref{TabBenchmarks}.

Figure~\ref{fig:nn_performance} shows the performance of the trained NN over 20 random samples of each of the six benchmarks. For each sample, we used the framework in~\cite{polyar} to compute the ground-truth percentage in the reduction of the volume of ambiguous regions after applying every action under-approximation, over-approximation, or split. We then evaluated the NN on each sample and reported in Figure~\ref{fig:nn_performance} both the ground-truth reduction of the ambiguous regions (as bars) against the index of the action suggested by the NN (as the text above the bars).
As it can be seen from Figure \ref{fig:nn_performance}, except for the second sample of the first benchmark, the NN outputs represent the actions that lead to the maximum reduction of the ambiguous region's volume.

\begin{table}[t!]
\caption{Evaluation of the trained NN on the six different benchmarks. 
}
\begin{adjustbox}{width=0.8\columnwidth,center}
\begin{tabular}{|c|c|c|c|c|c|}
    \hline
     Benchmark & $p\left(x\right)$ & $n$ & order & region & Accuracy  \\
    \hline
    \hline
   1  & $c_1x_1^2 + c_2x_2 + c_3x_1x_2 + c_4x_2 + c_5y + c_6$ & $2$ & $2$ &  $[-4, 4]^2$ &  $95\%$\\
    \hline 
   2 & $c_1x_1^2 + c_2x_2 + c_3x_1x_2 + c_4x_2 + c_5y + c_6$ & $2$ & $2$ & $[-10, 10]^2$  &  $93\%$ \\
     \hline 
   3  & $c_1x_1^4 + c_2x_2^3 + c_3x_1^4x_2^3 + c_4x_1^3 + c_5 $ & $2$ & $4$ & $[-2, 2]^2$ &  $88\%$ \\
   \hline
   4  & $c_1x_1^{10} + c_2x_2^5 + c_3x_1^5x_2^3 + c_4x_1^5 + c_5 $ & $2$ & $10$ & $[-2, 2]^2$ &  $87\%$ \\
    \hline 
   5  & $c_1x_1^3 + c_2x_2^3 + c_3x_3^3 + c_4x_4^3 $ & $4$ & $3$ & $[-2, 2]^4$ & $81\%$\\
      \hline 
    6  & $c_1x_1^3 + c_2x_2^3 + c_3x_3^3 + c_4x_4^3 + c_5x_5^3 + c_6x_6^3 + c_7x_7^3 $ & $7$ & $3$ & $[-2, 2]^7$ & $80\%$\\
      \hline
\end{tabular}
\end{adjustbox}
\label{TabBenchmarks}
\end{table}

Finally, we ran the same experiment for 1000 samples and report the percentage of samples for which the NN was able to predict the action that leads to the maximum reduction in the ambiguous region's volume. As it can be seen from Table \ref{TabBenchmarks}, the trained NN is able to generalize on the different benchmarks. For instance, evaluating the NN on different domains results in the lowest accuracy of $93\%$. Furthermore, evaluating the NN on polynomials with higher-order results in an accuracy of $87\%$. Finally, the NN achieves $80\%$ on higher dimension benchmarks.

\begin{figure*}[!t]
\centering
\resizebox{.9\textwidth}{!}{
    \centering
    \begin{tabular}{ c | c | c | c |}
         $m$ & SAT/UNSAT & \textbf{Execution times vs Polynomial Order} & \textbf{Execution times vs Number of Variables}\\\hline
         
         \multirow{0.5}{*}{}
         
         & 
         
         \multicolumn{3}{c|}{        \includegraphics[width=0.9\columnwidth, height = 0.04 \textheight]{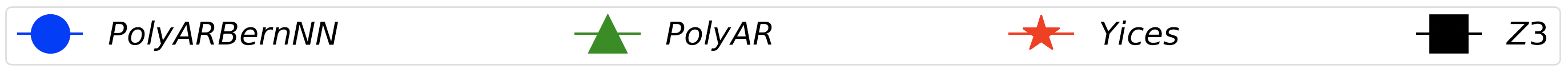}}

         \\ \hline
          \multirow{2}{*}{1} 
          & 
          \raisebox{5.0\totalheight}{UNSAT}
          &
         \includegraphics[width=0.65\columnwidth,trim=0mm 0 8mm 2mm, clip]{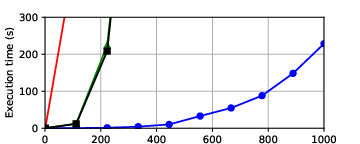}
         &
         \includegraphics[width=0.65\columnwidth,trim=0mm 0 8mm 2mm, clip]{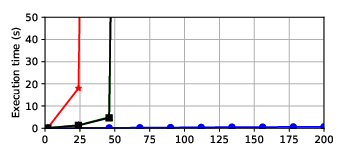}
         \\ \cline{2-4}
         %
          & 
          \raisebox{5.0\totalheight}{SAT}
          &
         \includegraphics[width=0.65\columnwidth,trim=0mm 0 8mm 2mm, clip]{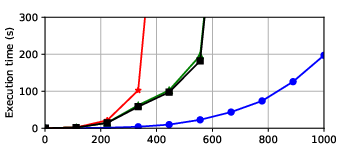}
         &
         \includegraphics[width=0.65\columnwidth,trim=0mm 0 8mm 2mm, clip]{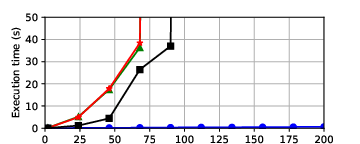}
        \\ \hline
         \multirow{2}{*}{5} 
          & 
          \raisebox{5.0\totalheight}{UNSAT}
          &
         \includegraphics[width=0.65\columnwidth,trim=0mm 0 8mm 2mm, clip]{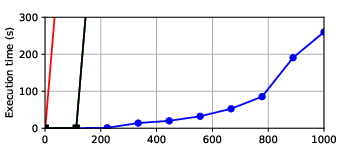}
         &
         \includegraphics[width=0.65\columnwidth,trim=0mm 0 8mm 2mm, clip]{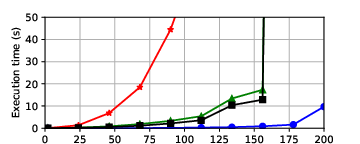}
         \\ \cline{2-4}
         %
          & 
          \raisebox{5.0\totalheight}{SAT}
          &
         \includegraphics[width=0.65\columnwidth,trim=0mm 0 8mm 2mm, clip]{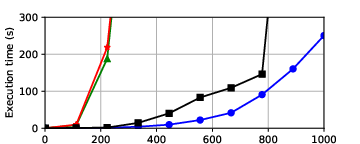}
         &
         \includegraphics[width=0.65\columnwidth,trim=0mm 0 8mm 2mm, clip]{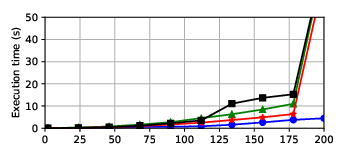}
         \\ \hline
         
          \multirow{2}{*}{10} 
          & 
          \raisebox{5.0\totalheight}{UNSAT}
          &
         \includegraphics[width=0.65\columnwidth,trim=0mm 0 8mm 2mm, clip]{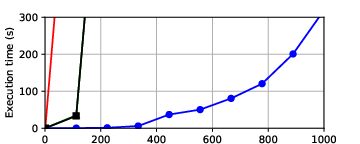}
         &
         \includegraphics[width=0.65\columnwidth,trim=0mm 0 8mm 2mm, clip]{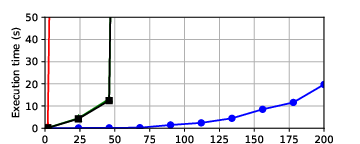}
         \\ \cline{2-4}
         %
          & 
          \raisebox{5.0\totalheight}{SAT}
          &
         \includegraphics[width=0.65\columnwidth,trim=0mm 0 8mm 2mm, clip]{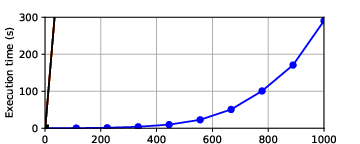}
         &
         \includegraphics[width=0.65\columnwidth,trim=0mm 0 8mm 2mm, clip]{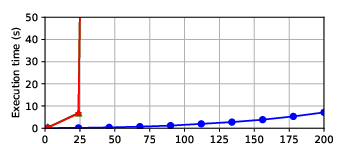}
         \\ \hline

         
    \end{tabular}
    }
    \caption{\small{Scalability results of PolyARBerNN in the UNSAT case for $1$, $5$, and $10$ constraints. (left) evolution of the execution time in seconds as a function of the order of the polynomials, (right) evolution of the execution time in seconds as a function of the number of variables. The timeout is equal to 1 hour.}}
    \label{figUNSAT}
    \vspace{-5mm}
\end{figure*}

\section{Numerical Results - Scalability Results}
In this section, we study the scalability of PolyARBerNN in terms of execution times by varying the order, the number of variables, and the number of the polynomial constraints for instances when a solution exists (the problem is Satisfiabile or SAT for short) and when a solution does not exist (or UNSAT for short). We will perform this study in comparison with state-of-the-art solvers including our previous solver PolyAR, Z3 8.9, and Yices 2.6.
Next, we compare the performance of PolyARBerNN against a theorem prover named PVS which implements a Bernstein library to solve multivariate polynomial constraints \cite{pvsnasa}. 
Finally, we compare the scalability of the PolyAROpt optimizer against the built-in optimization library in Z3 8.9 to solve an unconstrained multivariate polynomial optimization problem with varying order and number of variables.

\subsection{Scalability of PolyARBerNN against other solvers}

In this experiment, we compare the execution times of PolyARBerNN against the PolyAR tool~\cite{polyar}, Z3 8.9, and Yices 2.6. We consider two instances of Problem 1: an UNSAT and SAT problems. For each instance, we consider three scenarios, $m=1$, $m=5$, and $m=10$ where $m$ is the number of polynomial constraints. First, we vary the order of the polynomials from 0 to 1000 while fixing the number of variables (and hence the dimension of the search space) to two. Alternatively, we also fix the order of the polynomials to $30$ while varying the number of variables from 1 to 200.
%
We set the timeout of the simulations to be $1$ hour. Figure~\ref{figUNSAT} reports the execution times for all the experiments whenever the problem is UNSAT and SAT.
As evidenced by the figures, PolyARBerNN succeeded to solve the instances of Problem 1 for all orders and numbers of variables in a few seconds. For instance, solving $10$ polynomial constraints with $200$ variables and a maximum order of $30$ took around $20~s$ leading to a speed-up of $200\times$ compared to Z3 and Yices. On the other hand, other solvers are incapable of solving the polynomial constraints for all orders or number of variables and they time out after one hour.
These results show the scalability of the proposed approach by including Bernstein coefficients to prune the search space and a NN to guide the abstraction refinement.


\subsection{Scalability of PolyARBerNN versus PVS}
In this experiment, we want to investigate the effect of the NN on the overall performance of the system. 
We compare PolyARBerNN against the PVS theorem prover which also uses Bernstein coefficients to reason about polynomial constraints. The polynomials and the domains of the associated variables are given in~\cite{pvsnasa} which were originally used to assess the scalability of the PVS theorem prover. As evinced by the results in Table~\ref{tab:nasalib}, using a combination of Bernstein and adding the NN to guide the convex abstractions (in PolyARBerNN) leads to additional savings in the execution time, leading to a speedup of $483\times$ in the Heart dipole example.

\begin{table}[t!]
\caption{Scalability results for PolyARBerNN against PVS.}
\label{tab:nasalib}
\begin{adjustbox}{width=0.4\columnwidth,center}
\begin{tabular}{|c|c|c|c|}
    \hline
     Problem  & PVS & \multicolumn{2}{c|}{PolyARBerNN}\\
     \cline{3-4}
     & & Times (sec) & Speedup\\
    \hline
    \hline
    Schwefel  & $1.23 $ &   $\mathbf{0.022}$ & $54.9 \times$ \\
   \hline
   Reaction diffusion   & $0.25$ &  $\mathbf{0.018}$ & $12.88 \times$ \\
    \hline 
   Caprasse  & $0.03$ &  $\mathbf{0.025}$ & $0.19 \times$  \\
     \hline 
   Lotka-Volterra   & $0.22$ &  $\mathbf{0.026}$ & $7.46 \times$ \\
    \hline 
   Butcher   & $0.46$ &  $\mathbf{0.029}$ & $14.86 \times$\\
      \hline 
    Magnetism   & $1.82$ &  $\mathbf{0.028}$ & $64.0 \times$ \\
      \hline
    Heart dipole   & $15.01$ &  $\mathbf{0.031}$ & $483.19 \times$ \\
      \hline
\end{tabular}
\end{adjustbox}
\vspace{-2mm}
\end{table}

\begin{figure}[!b]
    \centering
	\includegraphics[width=0.48\textwidth]{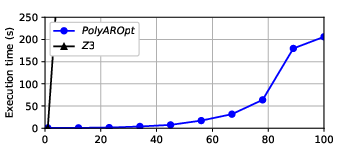} 
	\includegraphics[width=0.48\textwidth]{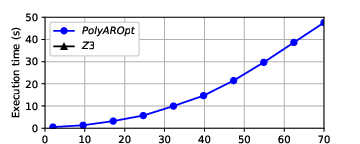} 
	\caption{Scalability results of PolyAROpt for unconstrained optimization. (left) evolution of the execution time in seconds as a function of the order of the polynomial. (right) evolution of the execution time in seconds as a function of the number of variables. The timeout is equal to 1 hour. 
	}
	\label{optimiz}
\end{figure}

\subsection{Scalability of PolyAROpt against other solvers}
We compare the scalability results of PolyAROpt with the Z3 solver due to the fact that Z3 has a built-in optimization library. Unfortunately, Yices does not have such an optimizer. We set the timeout of the simulation to be 1 hour. Figure \ref{optimiz} reports the execution times of two experiments of computing the minimum and maximum of unconstrained optimization. As evidenced by the two figures, PolyAROpt succeeded to solve the unconstrained optimization problem for all orders and number of variables. For instance, solving the unconstrained optimization problem with $70$ variables and a maximum order of $3$ took around $50~seconds$. On the other hand, the Z3 solver is incapable of solving the unconstrained optimization problem for all orders or number of variables and times out.

\section{Numerical Results - Use Cases}
In this section, we provide two engineering use cases. The first one focuses on the use of PolyARBerNN to synthesize stabilizing non-parametric controllers for nonlinear dynamical systems. The second use case focuses on the use of PolyAROPT to perform reachability analysis of polynomial dynamical systems.

\subsection{Use Case 1: Nonlinear Controller Design for a Duffing Oscillator}
In this subsection, we assess the scalability of the PolyARBerNN solver compared to state-of-the-art solvers for synthesizing a non-parametric controller for a Duffing oscillator reported by~\cite{MPCdesign}. All the details of the dynamics of the oscillator and how we generated the polynomial constraints can be found in \cite{polyar}. We denote by $n$ the dimension of the Duffing oscillator.
We consider two instances of the controller synthesis problem for the Duffing oscillator with the following parameters:
\begin{itemize}
    \item $n=3$, $\zeta=1.0$, $x\left(0\right)=[0.15,0.15,0.15]$, $L_1\left(x\left(k\right),u\left(k\right)\right)=\left(-x_1^{3}\left(k\right)+x_3^{3}\left(k\right)+ u\left(k\right) - 2 \right)^{51}$, $L_2\left(x\left(k\right),u\left(k\right)\right)=x_1^{51}\left(k\right)x_3^7\left(k\right) + x_1^9\left(k\right)x_3^5\left(k\right) - 5  x_2^4\left(k\right) - x_2^2\left(k\right)u^2\left(k\right)$, which results in $9$ polynomial constraints with $4$ variables and max polynomial order of $153$. 
    \item $n=4$, $\zeta=1.75$, $x\left(0\right)=[0.1,0.1,0.01, 0.1]$, $L_1\left(x\left(k\right),u\left(k\right)\right)=x_1^{4}\left(k\right)+x_2^{4}\left(k\right)+x_3^{4}\left(k\right)+x_4^{4}\left(k\right)-u^{4}\left(k\right)$, $L_2\left(x\left(k\right),u\left(k\right)\right)=-x_1^{51}\left(k\right)x_3^{20}\left(k\right) - 5x_2^4\left(k\right) - x_2^2\left(k\right)u^2\left(k\right)$, $L_3\left(x\left(k\right),u\left(k\right)\right)=\left(x_1x_2^2 - u\left(k\right) - 100\right)^{41}$, which results in $12$ polynomial constraints with $5$ variables and max polynomial order of $82$.
\end{itemize}

\begin{figure}
\resizebox{.75\textwidth}{!}{
    \centering
    \begin{tabular}{ c | c | c |}
         $n$ & State Space & Execution Time Evolution over time\\\hline
           3 & 
         \raisebox{-0.5\totalheight}{\includegraphics[width=0.6\columnwidth,trim=8mm 0 15mm 0, clip]{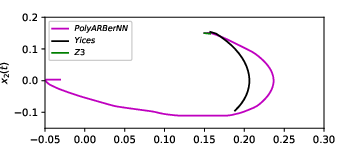}} &
         \raisebox{-0.5\totalheight}{\includegraphics[width=0.6\columnwidth,trim=8mm 0 15mm 0, clip]{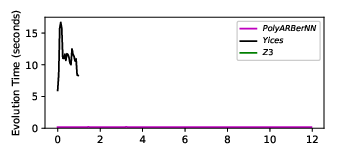}} \\ \hline
          4 & 
         \raisebox{-0.5\totalheight}{\includegraphics[width=0.6\columnwidth,trim=8mm 0 15mm 0, clip]{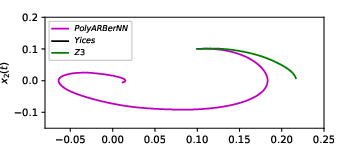}} &
         \raisebox{-0.5\totalheight}{\includegraphics[width=0.6\columnwidth,trim=8mm 0 15mm 0, clip]{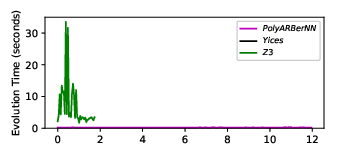}} \\ \hline
    \end{tabular}
    }
    \caption{Results of controlling the Duffing oscillator with different $n$ (left) evolution of the states $x_1(k)$ and $x_2(k)$ for the solvers in the state-space, (right) evolution of the execution time of solvers during the $12$ seconds. The timeout is equal to $60s$. Trajectories are truncated once the solver exceeds the timeout limit.}
    \label{fig:duff}
    \vspace{-3mm}
\end{figure}

We feed the resultant polynomial inequality constraint to PolyARBerNN, Yices, and Z3. We solve the feasibility problem for $n=3$ and $n=4$. We set the timeout to be $60s$. Figure~\ref{fig:duff} (left) shows the state-space evolution of the controlled Duffing oscillator for different solvers for number of variables $n$ of $3$ and $4$. Figure~\ref{fig:duff} (right) shows the evolution of the execution time of the solvers during the $12$ seconds. As it can be seen from Figure 7, our solver PolyARBerNN succeeded to find a control input $u$ that regulates the state to the origin for all $n$. However, off-the-shelf solvers are incapable of solving all two instances, and they early time out after $60~seconds$ out of the simulated $12~seconds$. 
This shows the scalability of the proposed approach.

\begin{figure*}[b]
    \centering
	\includegraphics[width=0.45\textwidth, height=0.2\textheight]{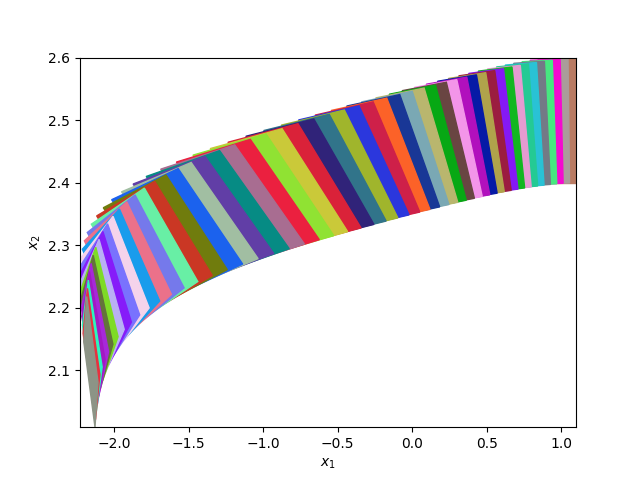} 
	\includegraphics[width=0.45\textwidth, height=0.2\textheight]{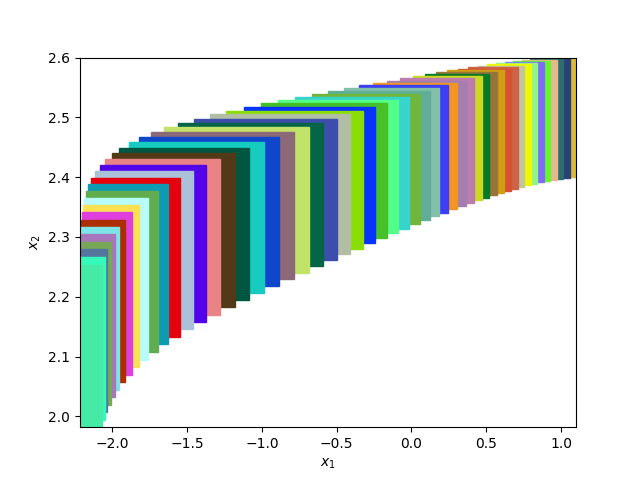}
 	\vspace{-3mm}
	\caption{Reachability computation for the FitzHugh-Nagumo neuron model. \textbf{Left:} using PolyAROpt for number of steps $K = 50$. \textbf{Right:} using Sapo for number of steps $K = 50$.
	}
	\label{fig:nagumo}
\end{figure*}

\begin{figure*}[!t]
    \centering
    \includegraphics[width=0.32\textwidth, height=0.2\textheight]{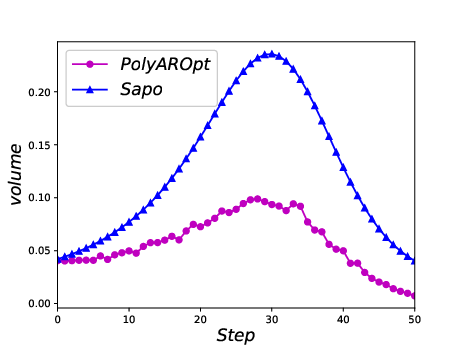}
	\includegraphics[width=0.32\textwidth, height=0.2\textheight]{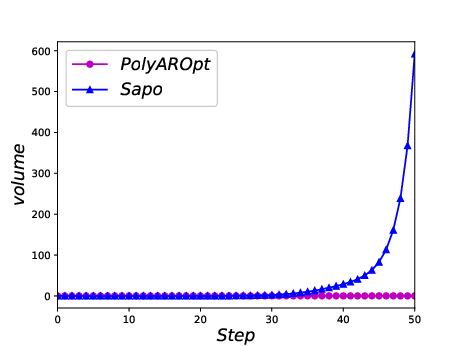} 
	\includegraphics[width=0.32\textwidth, height=0.2\textheight]{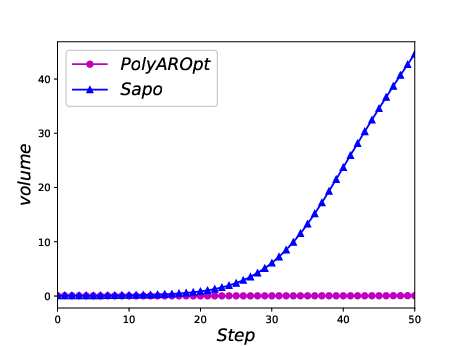} 
 	\vspace{-3mm}
	\caption{Volume of the reachable set of states that are obtained using PolyAROpt and Sapo (left) FitzHugh-Nagumo, (center) Duffing oscillator, and (right) Jet flight.
	}
	\label{fig:volumes}
\end{figure*}

\subsection{Use Case 2: Reachability analysis of a discrete polynomial dynamical systems}
In this section, we show how to use PolyAROpt to compute the reachable set of states for discrete-time polynomial systems. We consider a discrete polynomial dynamical system of the following form:
\begin{align}\label{polydynsys}
    x_{k + 1} = f\left(x_k\right),~k \in \mathbb{N}, x_0 \in Q_0,
\end{align}
where $f:\mathbb{R}^n \rightarrow \mathbb{R}^n$ is a multivariate polynomial map of a maximum degree $L = \left(l_1, \cdots, l_n\right)$, and $Q_0$ is a bounded polyhedron in $\mathbb{R}^n$. In this subsection, we consider a bounded-time reachability analysis of the system in~\eqref{polydynsys}. Computing the exact reachability sets of this type of dynamic system is hard. Therefore, we overapproximate the exact set with a simplified set such as bounded polyhedra. Bounded polyhedra are easy to handle and analyze. Computing the reachability set, after a finite time $K$, involves computing sequentially the reachability set at every time step $k$  using the following relation $Q_{k+1} = f\left(Q_k\right), k = 0, \cdots, K-1$, where $Q_k$ is a bounded polyhedra. 
A bounded polyhedra is represented with an H-representation $Q_k = \left(A_k, b_k\right) = \{ x \in \mathbb{R}^n | A_k x \leq b_k\}$, where the inequality is a point-wise inequality. The template $A_k$ represents the directions of $Q_k$'s faces and $b_k$ represents their positions. Given a polyhedra $Q_k = \left(A_k, b_k\right)$, we need to compute $Q_{k+1} = f\left(Q_k\right)$. We assume that $A_k \in \mathbb{R}^{m \times n}$ is given. Now, we need to compute $b_{k+1} \in \mathbb{R}^m$ which can be obtained through the following optimization problem \cite{mohamedreach}:
\begin{align} \label{reachoptimiz}
   -b_{k + 1, i} \leq u_{k + 1, i} = \min\limits_{x \in Q_k} - A_{k + 1, i} f\left(x\right), \forall i = 1, \cdots, m,
\end{align}
where $A_{k + 1, i}$ and $b_{k + 1, i}$ are the $i^{th}$ row and component of the templates $A_{k + 1}$ and $b_{k + 1}$, respectively. In every step $k \in \mathbb{N}$, we compute an upper bound for $-b_{k + 1, i}$, by solving the optimization problem \eqref{reachoptimiz} using PolyAROpt and then an overapproximation of the reachability set $Q_{k + 1}$ is computed. In order to use PolyAROpt to solve \eqref{reachoptimiz}, we need to overapproximate the polyhedra $Q_k$ with a hyperreactangle $R_k$. Therefore, the optimization problem \eqref{reachoptimiz} is modified as follows:
\begin{align} \label{reachoptimizmod}
   -b_{k + 1, i} \leq u_{k + 1, i} = &\min\limits_{x \in R_k} - A_{k + 1, i} f\left(x\right), \forall i = 1, \cdots, m,\nonumber \\
   &\text{subject to}~x~\in~Q_k.
\end{align}

We implemented the reachability computation method and tested it on three dynamical systems:
\begin{itemize}
    \item FitzHugh-Nagumo Neuron: is a polynomial dynamic systems that models the electrical activity of a neuron~\cite{FitzHughNagumo}. We performed reachability analysis for $K = 50$ time steps with the initial set of states $Q_0 = [0.9, 1.1] \times [2.4, 2.6]$,
    \item Duffing Oscillator: is a discrete-time version of a nonlinear oscillator model~\cite{polyar}. We performed reachability analysis for $K = 50$ time steps with the initial set of states $Q_0 = [2.49, 2.51] \times [1.49, 1.51]$,
    \item Jet Flight: is a discrete-time version of a jet flight model~\cite{jetflight}. We performed reachability analysis for $K = 50$ time steps with the initial set of states $Q_0 = [0.9, 1.2] \times [0.9, 1.2]$.
\end{itemize}
All the details of the dynamics of the three dynamical systems and how we generated the polynomial constraints can be found in \cite{mohamedreach, FitzHughNagumo, polyar, jetflight}. We computed the reachable sets for these dynamical systems and compared our results with Sapo \cite{sapo}. Sapo is a tool proposed for the reachability analysis of the class of discrete-time polynomial dynamical systems. Sapo linearizes the optimization problem~\eqref{reachoptimiz} using the Bernstein form of the polynomial and was shown to outperform state-of-the-art reachability analysis tools like Flow$^*$~\cite{chen2013flow}. Figure~\ref{fig:nagumo} shows the reachable sets computed by PolyAROpt compared to Sapo for the  FitzHugh-Nagumo Neuron model. Inspecting the results in Figure~\ref{fig:nagumo} qualitatively shows that PolyAROpt is capable of computing tighter sets. To quantitatively compare the results of PolyAROpt and Sapo, we compute the volume of each reachable set (for different time steps). Figure~\ref{fig:volumes} shows these volumes for all the three dynamical systems mentioned above.
As evident by the results in Figure~\ref{fig:volumes}, PolyAROpt results in reachable sets that are tighter than the one obtained from the Sapo thanks to PolyAROpt's ability of solving the polynomial optimization problem without any relaxations. Such ability to avoid relaxation results in several orders of magnitude reduction in the volume of the reachable sets compared to Sapo; a significant improvement in the analysis of such dynamical systems.

~\\
\noindent \textbf{Conclusions.} In this paper, we proposed PolyARBerNN, a solver for polynomial inequality constraints. We proposed a systematic methodology to design neural networks that can be used to guide the abstraction refinement process by bridging the gap between neural network properties and the properties of polynomial representations. We showed that the use of Bernstein coefficients leads the way to designing better neural network guides and provides an additional abstraction that can be used to accelerate the solver. We generalized the solver to reason about optimization problems. We demonstrated that the proposed solver outperforms state-of-the-art tools by several orders of magnitude.

\begin{acks} 
This work was supported by the National Science Foundation under grant numbers \#2002405 and \#2139781.
\end{acks}

\bibliographystyle{ACM-Reference-Format}
\bibliography{bibliography}


\begin{thebibliography}{54}


\ifx \showCODEN    \undefined \def \showCODEN     #1{\unskip}     \fi
\ifx \showDOI      \undefined \def \showDOI       #1{#1}\fi
\ifx \showISBNx    \undefined \def \showISBNx     #1{\unskip}     \fi
\ifx \showISBNxiii \undefined \def \showISBNxiii  #1{\unskip}     \fi
\ifx \showISSN     \undefined \def \showISSN      #1{\unskip}     \fi
\ifx \showLCCN     \undefined \def \showLCCN      #1{\unskip}     \fi
\ifx \shownote     \undefined \def \shownote      #1{#1}          \fi
\ifx \showarticletitle \undefined \def \showarticletitle #1{#1}   \fi
\ifx \showURL      \undefined \def \showURL       {\relax}        \fi
\providecommand\bibfield[2]{#2}
\providecommand\bibinfo[2]{#2}
\providecommand\natexlab[1]{#1}
\providecommand\showeprint[2][]{arXiv:#2}

\bibitem[mor(2018)]%
        {morecite2}
 \bibinfo{year}{2018}\natexlab{}.
\newblock \showarticletitle{{Under-approximating reach sets for polynomial
  continuous systems}, author={Xue, Bai and Fr{\"a}nzle, Martin and Zhan,
  Naijun}}. In \bibinfo{booktitle}{\emph{Proceedings of the 21st International
  Conference on Hybrid Systems: Computation and Control (part of CPS Week)}}.
  \bibinfo{pages}{51--60}.
\newblock


\bibitem[An et~al\mbox{.}(2018)]%
        {morecite3}
\bibfield{author}{\bibinfo{person}{Jie An}, \bibinfo{person}{Naijun Zhan},
  \bibinfo{person}{Xiaoshan Li}, \bibinfo{person}{Miaomiao Zhang}, {and}
  \bibinfo{person}{Wang Yi}.} \bibinfo{year}{2018}\natexlab{}.
\newblock \showarticletitle{{Model checking bounded continuous-time extended
  linear duration invariants}}. In \bibinfo{booktitle}{\emph{Proceedings of the
  21st International Conference on Hybrid Systems: Computation and Control
  (part of CPS Week)}}. \bibinfo{pages}{81--90}.
\newblock


\bibitem[Bak et~al\mbox{.}(2019)]%
        {morecite5}
\bibfield{author}{\bibinfo{person}{Stanley Bak}, \bibinfo{person}{Hoang-Dung
  Tran}, {and} \bibinfo{person}{Taylor~T Johnson}.}
  \bibinfo{year}{2019}\natexlab{}.
\newblock \showarticletitle{{Numerical verification of affine systems with up
  to a billion dimensions}}. In \bibinfo{booktitle}{\emph{Proceedings of the
  22nd ACM International Conference on Hybrid Systems: Computation and
  Control}}. \bibinfo{pages}{23--32}.
\newblock


\bibitem[Bello et~al\mbox{.}(2016)]%
        {travelsales}
\bibfield{author}{\bibinfo{person}{Irwan Bello}, \bibinfo{person}{Hieu Pham},
  \bibinfo{person}{Quoc~V Le}, \bibinfo{person}{Mohammad Norouzi}, {and}
  \bibinfo{person}{Samy Bengio}.} \bibinfo{year}{2016}\natexlab{}.
\newblock \showarticletitle{{Neural combinatorial optimization with
  reinforcement learning}}.
\newblock \bibinfo{journal}{\emph{arXiv preprint arXiv:1611.09940}}
  (\bibinfo{year}{2016}).
\newblock


\bibitem[Ben~Sassi et~al\mbox{.}(2012)]%
        {mohamedreach}
\bibfield{author}{\bibinfo{person}{Mohamed~Amin Ben~Sassi},
  \bibinfo{person}{Romain Testylier}, \bibinfo{person}{Thao Dang}, {and}
  \bibinfo{person}{Antoine Girard}.} \bibinfo{year}{2012}\natexlab{}.
\newblock \showarticletitle{{Reachability analysis of polynomial systems using
  linear programming relaxations}}. In \bibinfo{booktitle}{\emph{International
  Symposium on Automated Technology for Verification and Analysis}}. Springer,
  \bibinfo{pages}{137--151}.
\newblock


\bibitem[Boreale(2018)]%
        {morecite4}
\bibfield{author}{\bibinfo{person}{Michele Boreale}.}
  \bibinfo{year}{2018}\natexlab{}.
\newblock \showarticletitle{{Algorithms for exact and approximate linear
  abstractions of polynomial continuous systems}}. In
  \bibinfo{booktitle}{\emph{Proceedings of the 21st International Conference on
  Hybrid Systems: Computation and Control (part of CPS Week)}}.
  \bibinfo{pages}{207--216}.
\newblock


\bibitem[Brisebarre et~al\mbox{.}(2012)]%
        {coqtaylor}
\bibfield{author}{\bibinfo{person}{Nicolas Brisebarre}, \bibinfo{person}{Mioara
  Jolde{\c{s}}}, \bibinfo{person}{{\'E}rik Martin-Dorel},
  \bibinfo{person}{Micaela Mayero}, \bibinfo{person}{Jean-Michel Muller},
  \bibinfo{person}{Ioana Pa{\c{s}}ca}, \bibinfo{person}{Laurence Rideau}, {and}
  \bibinfo{person}{Laurent Th{\'e}ry}.} \bibinfo{year}{2012}\natexlab{}.
\newblock \showarticletitle{{Rigorous polynomial approximation using Taylor
  models in Coq}}. In \bibinfo{booktitle}{\emph{NASA Formal Methods
  Symposium}}. Springer, \bibinfo{pages}{85--99}.
\newblock


\bibitem[Brown(2001)]%
        {brown}
\bibfield{author}{\bibinfo{person}{Christopher~W Brown}.}
  \bibinfo{year}{2001}\natexlab{}.
\newblock \showarticletitle{{Improved projection for cylindrical algebraic
  decomposition}}.
\newblock \bibinfo{journal}{\emph{Journal of Symbolic Computation}}
  \bibinfo{volume}{32}, \bibinfo{number}{5} (\bibinfo{year}{2001}),
  \bibinfo{pages}{447--465}.
\newblock


\bibitem[Chen(2015)]%
        {jetflight}
\bibfield{author}{\bibinfo{person}{Xin Chen}.} \bibinfo{year}{2015}\natexlab{}.
\newblock \emph{\bibinfo{title}{{Reachability analysis of non-linear hybrid
  systems using taylor models}}}.
\newblock \bibinfo{thesistype}{Ph.\,D. Dissertation}.
  \bibinfo{school}{Fachgruppe Informatik, RWTH Aachen University}.
\newblock


\bibitem[Chen et~al\mbox{.}(2013)]%
        {chen2013flow}
\bibfield{author}{\bibinfo{person}{Xin Chen}, \bibinfo{person}{Erika
  {\'A}brah{\'a}m}, {and} \bibinfo{person}{Sriram Sankaranarayanan}.}
  \bibinfo{year}{2013}\natexlab{}.
\newblock \showarticletitle{Flow*: An analyzer for non-linear hybrid systems}.
  In \bibinfo{booktitle}{\emph{International Conference on Computer Aided
  Verification}}. Springer, \bibinfo{pages}{258--263}.
\newblock


\bibitem[Collins(1975)]%
        {collins}
\bibfield{author}{\bibinfo{person}{George~E Collins}.}
  \bibinfo{year}{1975}\natexlab{}.
\newblock \showarticletitle{{Quantifier elimination for real closed fields by
  cylindrical algebraic decomposition}}. In \bibinfo{booktitle}{\emph{Automata
  Theory and Formal Languages 2nd GI Conference Kaisers lautern, May 20--23,
  1975}}. Springer, \bibinfo{pages}{134--183}.
\newblock


\bibitem[De~Moura and Bj{\o}rner(2008)]%
        {Z3}
\bibfield{author}{\bibinfo{person}{Leonardo De~Moura} {and}
  \bibinfo{person}{Nikolaj Bj{\o}rner}.} \bibinfo{year}{2008}\natexlab{}.
\newblock \showarticletitle{{Z3: An efficient SMT solver}}. In
  \bibinfo{booktitle}{\emph{International conference on Tools and Algorithms
  for the Construction and Analysis of Systems}}. \bibinfo{pages}{337--340}.
\newblock


\bibitem[Devlin et~al\mbox{.}(2017)]%
        {progsyn}
\bibfield{author}{\bibinfo{person}{Jacob Devlin}, \bibinfo{person}{Jonathan
  Uesato}, \bibinfo{person}{Surya Bhupatiraju}, \bibinfo{person}{Rishabh
  Singh}, \bibinfo{person}{Abdel-rahman Mohamed}, {and}
  \bibinfo{person}{Pushmeet Kohli}.} \bibinfo{year}{2017}\natexlab{}.
\newblock \showarticletitle{{Robustfill: Neural program learning under noisy
  i/o}}. In \bibinfo{booktitle}{\emph{International conference on machine
  learning}}. PMLR, \bibinfo{pages}{990--998}.
\newblock


\bibitem[Dreossi(2017)]%
        {sapo}
\bibfield{author}{\bibinfo{person}{Tommaso Dreossi}.}
  \bibinfo{year}{2017}\natexlab{}.
\newblock \showarticletitle{{Sapo: Reachability computation and parameter
  synthesis of polynomial dynamical systems}}. In
  \bibinfo{booktitle}{\emph{Proceedings of the 20th International Conference on
  Hybrid Systems: Computation and Control}}. \bibinfo{pages}{29--34}.
\newblock


\bibitem[Dutertre(2014)]%
        {yices}
\bibfield{author}{\bibinfo{person}{Bruno Dutertre}.}
  \bibinfo{year}{2014}\natexlab{}.
\newblock \showarticletitle{{Yices 2.2}}. In
  \bibinfo{booktitle}{\emph{International Conference on Computer Aided
  Verification}}. Springer, \bibinfo{pages}{737--744}.
\newblock


\bibitem[Dutta et~al\mbox{.}(2019)]%
        {morecite7}
\bibfield{author}{\bibinfo{person}{Souradeep Dutta}, \bibinfo{person}{Xin
  Chen}, {and} \bibinfo{person}{Sriram Sankaranarayanan}.}
  \bibinfo{year}{2019}\natexlab{}.
\newblock \showarticletitle{{Reachability analysis for neural feedback systems
  using regressive polynomial rule inference}}. In
  \bibinfo{booktitle}{\emph{Proceedings of the 22nd ACM International
  Conference on Hybrid Systems: Computation and Control}}.
  \bibinfo{pages}{157--168}.
\newblock


\bibitem[England and Davenport(2016)]%
        {complexityproblem1}
\bibfield{author}{\bibinfo{person}{Matthew England} {and}
  \bibinfo{person}{James~H Davenport}.} \bibinfo{year}{2016}\natexlab{}.
\newblock \showarticletitle{{The complexity of cylindrical algebraic
  decomposition with respect to polynomial degree}}. In
  \bibinfo{booktitle}{\emph{International Workshop on Computer Algebra in
  Scientific Computing}}. Springer, \bibinfo{pages}{172--192}.
\newblock


\bibitem[Farouki and Rajan(1987)]%
        {faroukinumstab}
\bibfield{author}{\bibinfo{person}{Rida~T Farouki} {and} \bibinfo{person}{VT
  Rajan}.} \bibinfo{year}{1987}\natexlab{}.
\newblock \showarticletitle{{On the numerical condition of polynomials in
  Bernstein form}}.
\newblock \bibinfo{journal}{\emph{Computer Aided Geometric Design}}
  \bibinfo{volume}{4}, \bibinfo{number}{3} (\bibinfo{year}{1987}),
  \bibinfo{pages}{191--216}.
\newblock


\bibitem[Ferlez and Shoukry(2020)]%
        {TLL}
\bibfield{author}{\bibinfo{person}{James Ferlez} {and} \bibinfo{person}{Yasser
  Shoukry}.} \bibinfo{year}{2020}\natexlab{}.
\newblock \showarticletitle{{AReN: assured ReLU NN architecture for model
  predictive control of LTI systems}}. In \bibinfo{booktitle}{\emph{Proceedings
  of the 23rd International Conference on Hybrid Systems: Computation and
  Control}}. \bibinfo{pages}{1--11}.
\newblock


\bibitem[FitzHugh(1961)]%
        {FitzHughNagumo}
\bibfield{author}{\bibinfo{person}{Richard FitzHugh}.}
  \bibinfo{year}{1961}\natexlab{}.
\newblock \showarticletitle{{Impulses and physiological states in theoretical
  models of nerve membrane}}.
\newblock \bibinfo{journal}{\emph{Biophysical journal}} \bibinfo{volume}{1},
  \bibinfo{number}{6} (\bibinfo{year}{1961}), \bibinfo{pages}{445--466}.
\newblock


\bibitem[Fotiou et~al\mbox{.}(2006)]%
        {MPCdesign}
\bibfield{author}{\bibinfo{person}{Ioannis~A Fotiou}, \bibinfo{person}{Philipp
  Rostalski}, \bibinfo{person}{Pablo~A Parrilo}, {and} \bibinfo{person}{Manfred
  Morari}.} \bibinfo{year}{2006}\natexlab{}.
\newblock \showarticletitle{{Parametric optimization and optimal control using
  algebraic geometry methods}}.
\newblock \bibinfo{journal}{\emph{Internat. J. Control}} \bibinfo{volume}{79},
  \bibinfo{number}{11} (\bibinfo{year}{2006}), \bibinfo{pages}{1340--1358}.
\newblock


\bibitem[Garloff(1985a)]%
        {garloff}
\bibfield{author}{\bibinfo{person}{J{\"u}rgen Garloff}.}
  \bibinfo{year}{1985}\natexlab{a}.
\newblock \showarticletitle{{Convergent bounds for the range of multivariate
  polynomials}}. In \bibinfo{booktitle}{\emph{International Symposium on
  Interval Mathematics}}. Springer, \bibinfo{pages}{37--56}.
\newblock


\bibitem[Garloff(1985b)]%
        {range1}
\bibfield{author}{\bibinfo{person}{J{\"u}rgen Garloff}.}
  \bibinfo{year}{1985}\natexlab{b}.
\newblock \showarticletitle{{Convergent bounds for the range of multivariate
  polynomials}}. In \bibinfo{booktitle}{\emph{International Symposium on
  Interval Mathematics}}. Springer, \bibinfo{pages}{37--56}.
\newblock


\bibitem[Garloff and Smith(2004)]%
        {range2}
\bibfield{author}{\bibinfo{person}{J{\"u}rgen Garloff} {and}
  \bibinfo{person}{Andrew~P Smith}.} \bibinfo{year}{2004}\natexlab{}.
\newblock \showarticletitle{{An improved method for the computation of affine
  lower bound functions for polynomials}}.
\newblock In \bibinfo{booktitle}{\emph{Frontiers in Global Optimization}}.
  \bibinfo{publisher}{Springer}, \bibinfo{pages}{135--144}.
\newblock


\bibitem[Gasse et~al\mbox{.}(2019)]%
        {gcnn}
\bibfield{author}{\bibinfo{person}{Maxime Gasse}, \bibinfo{person}{Didier
  Chételat}, \bibinfo{person}{Nicola Ferroni}, \bibinfo{person}{Laurent
  Charlin}, {and} \bibinfo{person}{Andrea Lodi}.}
  \bibinfo{year}{2019}\natexlab{}.
\newblock \showarticletitle{{Exact Combinatorial Optimization with Graph
  Convolutional Neural Networks}}. In \bibinfo{booktitle}{\emph{Advances in
  Neural Information Processing Systems 32}}.
\newblock


\bibitem[Goldfarb and Idnani(1983)]%
        {quadprog}
\bibfield{author}{\bibinfo{person}{Donald Goldfarb} {and}
  \bibinfo{person}{Ashok Idnani}.} \bibinfo{year}{1983}\natexlab{}.
\newblock \showarticletitle{{A numerically stable dual method for solving
  strictly convex quadratic programs}}.
\newblock \bibinfo{journal}{\emph{Mathematical programming}}
  \bibinfo{volume}{27}, \bibinfo{number}{1} (\bibinfo{year}{1983}),
  \bibinfo{pages}{1--33}.
\newblock


\bibitem[Hong(1990)]%
        {Hong}
\bibfield{author}{\bibinfo{person}{H. Hong}.} \bibinfo{year}{1990}\natexlab{}.
\newblock \showarticletitle{{An Improvement of the Projection Operator in
  Cylindrical Algebraic Decomposition}}. In
  \bibinfo{booktitle}{\emph{Proceedings of the International Symposium on
  Symbolic and Algebraic Computation}} (Tokyo, Japan)
  \emph{(\bibinfo{series}{ISSAC '90})}. \bibinfo{publisher}{Association for
  Computing Machinery}, \bibinfo{address}{New York, NY, USA},
  \bibinfo{pages}{261–264}.
\newblock
\showISBNx{0201548925}
\urldef\tempurl%
\url{https://doi.org/10.1145/96877.96943}
\showDOI{\tempurl}


\bibitem[Irving et~al\mbox{.}(2016)]%
        {1thepro}
\bibfield{author}{\bibinfo{person}{Geoffrey Irving}, \bibinfo{person}{Christian
  Szegedy}, \bibinfo{person}{Alexander~A Alemi}, \bibinfo{person}{Niklas
  E{\'e}n}, \bibinfo{person}{Fran{\c{c}}ois Chollet}, {and}
  \bibinfo{person}{Josef Urban}.} \bibinfo{year}{2016}\natexlab{}.
\newblock \showarticletitle{{Deepmath-deep sequence models for premise
  selection}}.
\newblock \bibinfo{journal}{\emph{Advances in Neural Information Processing
  Systems}}  \bibinfo{volume}{29} (\bibinfo{year}{2016}),
  \bibinfo{pages}{2235--2243}.
\newblock


\bibitem[Kaliszyk et~al\mbox{.}(2017)]%
        {highthepro}
\bibfield{author}{\bibinfo{person}{Cezary Kaliszyk},
  \bibinfo{person}{Fran{\c{c}}ois Chollet}, {and} \bibinfo{person}{Christian
  Szegedy}.} \bibinfo{year}{2017}\natexlab{}.
\newblock \showarticletitle{{Holstep: A machine learning dataset for
  higher-order logic theorem proving}}.
\newblock \bibinfo{journal}{\emph{arXiv preprint arXiv:1703.00426}}
  (\bibinfo{year}{2017}).
\newblock


\bibitem[Kochdumper and Althoff(2020)]%
        {morecite9}
\bibfield{author}{\bibinfo{person}{Niklas Kochdumper} {and}
  \bibinfo{person}{Matthias Althoff}.} \bibinfo{year}{2020}\natexlab{}.
\newblock \showarticletitle{{Reachability analysis for hybrid systems with
  nonlinear guard sets}}. In \bibinfo{booktitle}{\emph{Proceedings of the 23rd
  International Conference on Hybrid Systems: Computation and Control}}.
  \bibinfo{pages}{1--10}.
\newblock


\bibitem[LeCun et~al\mbox{.}(2012)]%
        {lecun}
\bibfield{author}{\bibinfo{person}{Yann~A LeCun}, \bibinfo{person}{L{\'e}on
  Bottou}, \bibinfo{person}{Genevieve~B Orr}, {and}
  \bibinfo{person}{Klaus-Robert M{\"u}ller}.} \bibinfo{year}{2012}\natexlab{}.
\newblock \showarticletitle{{Efficient backprop}}.
\newblock In \bibinfo{booktitle}{\emph{Neural networks: Tricks of the trade}}.
  \bibinfo{publisher}{Springer}, \bibinfo{pages}{9--48}.
\newblock


\bibitem[Lestingi et~al\mbox{.}(2020)]%
        {robotverif}
\bibfield{author}{\bibinfo{person}{Livia Lestingi}, \bibinfo{person}{Mehrnoosh
  Askarpour}, \bibinfo{person}{Marcello~M Bersani}, {and}
  \bibinfo{person}{Matteo Rossi}.} \bibinfo{year}{2020}\natexlab{}.
\newblock \showarticletitle{{Formal verification of human-robot interaction in
  healthcare scenarios}}. In \bibinfo{booktitle}{\emph{International Conference
  on Software Engineering and Formal Methods}}. Springer,
  \bibinfo{pages}{303--324}.
\newblock


\bibitem[Mahboubi(2006)]%
        {coq}
\bibfield{author}{\bibinfo{person}{Assia Mahboubi}.}
  \bibinfo{year}{2006}\natexlab{}.
\newblock \showarticletitle{{Programming and certifying a CAD algorithm in the
  Coq system}}. In \bibinfo{booktitle}{\emph{Dagstuhl Seminar Proceedings}}.
  Schloss Dagstuhl-Leibniz-Zentrum f{\"u}r Informatik.
\newblock


\bibitem[Manual(1987)]%
        {cplex}
\bibfield{author}{\bibinfo{person}{CPLEX~User’s Manual}.}
  \bibinfo{year}{1987}\natexlab{}.
\newblock \showarticletitle{{Ibm ilog cplex optimization studio}}.
\newblock \bibinfo{journal}{\emph{Version}}  \bibinfo{volume}{12}
  (\bibinfo{year}{1987}), \bibinfo{pages}{1987--2018}.
\newblock


\bibitem[Marcucci and Tedrake(2019)]%
        {morecite8}
\bibfield{author}{\bibinfo{person}{Tobia Marcucci} {and} \bibinfo{person}{Russ
  Tedrake}.} \bibinfo{year}{2019}\natexlab{}.
\newblock \showarticletitle{{Mixed-integer formulations for optimal control of
  piecewise-affine systems}}. In \bibinfo{booktitle}{\emph{Proceedings of the
  22nd ACM International Conference on Hybrid Systems: Computation and
  Control}}. \bibinfo{pages}{230--239}.
\newblock


\bibitem[McCallum(1998)]%
        {McCalum}
\bibfield{author}{\bibinfo{person}{Scott McCallum}.}
  \bibinfo{year}{1998}\natexlab{}.
\newblock \showarticletitle{An improved projection operation for cylindrical
  algebraic decomposition}.
\newblock In \bibinfo{booktitle}{\emph{Quantifier Elimination and Cylindrical
  Algebraic Decomposition}}. \bibinfo{publisher}{Springer},
  \bibinfo{pages}{242--268}.
\newblock


\bibitem[Munoz and Narkawicz(2013)]%
        {pvsnasa}
\bibfield{author}{\bibinfo{person}{C{\'e}sar Munoz} {and}
  \bibinfo{person}{Anthony Narkawicz}.} \bibinfo{year}{2013}\natexlab{}.
\newblock \showarticletitle{{Formalization of Bernstein polynomials and
  applications to global optimization}}.
\newblock \bibinfo{journal}{\emph{Journal of Automated Reasoning}}
  \bibinfo{volume}{51}, \bibinfo{number}{2} (\bibinfo{year}{2013}),
  \bibinfo{pages}{151--196}.
\newblock


\bibitem[Munoz(2015)]%
        {airtraffic}
\bibfield{author}{\bibinfo{person}{Cesar~A Munoz}.}
  \bibinfo{year}{2015}\natexlab{}.
\newblock \showarticletitle{{Formal Methods in Air Traffic Management: The Case
  of Unmanned Aircraft Systems}}. In \bibinfo{booktitle}{\emph{International
  Colloquium on Theoretical Aspects of Computing (ICTAC 2015)}}.
\newblock


\bibitem[Narkawicz and Munoz(2012)]%
        {verifaircraft}
\bibfield{author}{\bibinfo{person}{Anthony Narkawicz} {and}
  \bibinfo{person}{C{\'e}sar~A Munoz}.} \bibinfo{year}{2012}\natexlab{}.
\newblock \showarticletitle{{Formal Verification of Conflict Detection
  Algorithms for Arbitrary Trajectories.}}
\newblock \bibinfo{journal}{\emph{Reliab. Comput.}}  \bibinfo{volume}{17}
  (\bibinfo{year}{2012}), \bibinfo{pages}{209--237}.
\newblock


\bibitem[Nelson and Pecheur(2002)]%
        {spaceshuttle}
\bibfield{author}{\bibinfo{person}{Stacy~D Nelson} {and}
  \bibinfo{person}{Charles Pecheur}.} \bibinfo{year}{2002}\natexlab{}.
\newblock \showarticletitle{{Formal verification for a next-generation space
  shuttle}}. In \bibinfo{booktitle}{\emph{International Workshop on Formal
  Approaches to Agent-Based Systems}}. Springer, \bibinfo{pages}{53--67}.
\newblock


\bibitem[Rabi(2020)]%
        {morecite11}
\bibfield{author}{\bibinfo{person}{Maben Rabi}.}
  \bibinfo{year}{2020}\natexlab{}.
\newblock \showarticletitle{{Piece-wise analytic trajectory computation for
  polytopic switching between stable affine systems}}. In
  \bibinfo{booktitle}{\emph{Proceedings of the 23rd International Conference on
  Hybrid Systems: Computation and Control}}. \bibinfo{pages}{1--11}.
\newblock


\bibitem[Ray and Nataraj(2012)]%
        {berncomplex}
\bibfield{author}{\bibinfo{person}{Shashwati Ray} {and} \bibinfo{person}{PSV
  Nataraj}.} \bibinfo{year}{2012}\natexlab{}.
\newblock \showarticletitle{{A Matrix Method for Efficient Computation of
  Bernstein Coefficients.}}
\newblock \bibinfo{journal}{\emph{Reliab. Comput.}} \bibinfo{volume}{17},
  \bibinfo{number}{1} (\bibinfo{year}{2012}), \bibinfo{pages}{40--71}.
\newblock


\bibitem[Roy et~al\mbox{.}(2018)]%
        {cpssynthesis}
\bibfield{author}{\bibinfo{person}{Debayan Roy}, \bibinfo{person}{Michael
  Balszun}, \bibinfo{person}{Thomas Heurung}, {and} \bibinfo{person}{Samarjit
  Chakraborty}.} \bibinfo{year}{2018}\natexlab{}.
\newblock \showarticletitle{{Multi-Domain Coupling for Automated Synthesis of
  Distributed Cyber-Physical Systems}}. In \bibinfo{booktitle}{\emph{2018 IEEE
  International Symposium on Circuits and Systems (ISCAS)}}. IEEE,
  \bibinfo{pages}{1--5}.
\newblock


\bibitem[Sadraddini and Tedrake(2020)]%
        {morecite10}
\bibfield{author}{\bibinfo{person}{Sadra Sadraddini} {and}
  \bibinfo{person}{Russ Tedrake}.} \bibinfo{year}{2020}\natexlab{}.
\newblock \showarticletitle{{Robust output feedback control with guaranteed
  constraint satisfaction}}. In \bibinfo{booktitle}{\emph{Proceedings of the
  23rd International Conference on Hybrid Systems: Computation and Control}}.
  \bibinfo{pages}{1--10}.
\newblock


\bibitem[Selsam et~al\mbox{.}(2019)]%
        {satsolve}
\bibfield{author}{\bibinfo{person}{Daniel Selsam}, \bibinfo{person}{Matthew
  Lamm}, \bibinfo{person}{Benedikt B\"{u}nz}, \bibinfo{person}{Percy Liang},
  \bibinfo{person}{Leonardo de Moura}, {and} \bibinfo{person}{David~L. Dill}.}
  \bibinfo{year}{2019}\natexlab{}.
\newblock \showarticletitle{{Learning a {SAT} Solver from Single-Bit
  Supervision}}. In \bibinfo{booktitle}{\emph{International Conference on
  Learning Representations}}.
\newblock
\urldef\tempurl%
\url{https://openreview.net/forum?id=HJMC_iA5tm}
\showURL{%
\tempurl}


\bibitem[Shen et~al\mbox{.}(2019)]%
        {NNcomplexity1}
\bibfield{author}{\bibinfo{person}{Zuowei Shen}, \bibinfo{person}{Haizhao
  Yang}, {and} \bibinfo{person}{Shijun Zhang}.}
  \bibinfo{year}{2019}\natexlab{}.
\newblock \showarticletitle{{Deep network approximation characterized by number
  of neurons}}.
\newblock \bibinfo{journal}{\emph{arXiv preprint arXiv:1906.05497}}
  (\bibinfo{year}{2019}).
\newblock


\bibitem[Shoukry et~al\mbox{.}(2018)]%
        {cpsdesign}
\bibfield{author}{\bibinfo{person}{Yasser Shoukry}, \bibinfo{person}{Michelle
  Chong}, \bibinfo{person}{Masashi Wakaiki}, \bibinfo{person}{Pierluigi Nuzzo},
  \bibinfo{person}{Alberto Sangiovanni-Vincentelli}, \bibinfo{person}{Sanjit~A
  Seshia}, \bibinfo{person}{Joao~P Hespanha}, {and} \bibinfo{person}{Paulo
  Tabuada}.} \bibinfo{year}{2018}\natexlab{}.
\newblock \showarticletitle{{SMT-based observer design for cyber-physical
  systems under sensor attacks}}.
\newblock \bibinfo{journal}{\emph{ACM Transactions on Cyber-Physical Systems}}
  \bibinfo{volume}{2}, \bibinfo{number}{1} (\bibinfo{year}{2018}),
  \bibinfo{pages}{1--27}.
\newblock


\bibitem[Smith(2009)]%
        {range4}
\bibfield{author}{\bibinfo{person}{Andrew~Paul Smith}.}
  \bibinfo{year}{2009}\natexlab{}.
\newblock \showarticletitle{{Fast construction of constant bound functions for
  sparse polynomials}}.
\newblock \bibinfo{journal}{\emph{Journal of Global Optimization}}
  \bibinfo{volume}{43}, \bibinfo{number}{2} (\bibinfo{year}{2009}),
  \bibinfo{pages}{445--458}.
\newblock


\bibitem[Sun et~al\mbox{.}(2019)]%
        {autosysverif}
\bibfield{author}{\bibinfo{person}{Xiaowu Sun}, \bibinfo{person}{Haitham
  Khedr}, {and} \bibinfo{person}{Yasser Shoukry}.}
  \bibinfo{year}{2019}\natexlab{}.
\newblock \showarticletitle{{Formal verification of neural network controlled
  autonomous systems}}. In \bibinfo{booktitle}{\emph{Proceedings of the 22nd
  ACM International Conference on Hybrid Systems: Computation and Control}}.
  \bibinfo{pages}{147--156}.
\newblock


\bibitem[Vandenberghe(2010)]%
        {cvxopt}
\bibfield{author}{\bibinfo{person}{Lieven Vandenberghe}.}
  \bibinfo{year}{2010}\natexlab{}.
\newblock \showarticletitle{{The CVXOPT linear and quadratic cone program
  solvers}}.
\newblock \bibinfo{journal}{\emph{Online: http://cvxopt.
  org/documentation/coneprog. pdf}} (\bibinfo{year}{2010}).
\newblock


\bibitem[Vinod and Oishi(2018)]%
        {morecite1}
\bibfield{author}{\bibinfo{person}{Abraham~P Vinod} {and}
  \bibinfo{person}{Meeko~MK Oishi}.} \bibinfo{year}{2018}\natexlab{}.
\newblock \showarticletitle{{Scalable underapproximative verification of
  stochastic LTI systems using convexity and compactness}}. In
  \bibinfo{booktitle}{\emph{Proceedings of the 21st International Conference on
  Hybrid Systems: Computation and Control (Part of CPS Week)}}.
  \bibinfo{pages}{1--10}.
\newblock


\bibitem[Xue et~al\mbox{.}(2019)]%
        {morecite6}
\bibfield{author}{\bibinfo{person}{Bai Xue}, \bibinfo{person}{Qiuye Wang},
  \bibinfo{person}{Naijun Zhan}, {and} \bibinfo{person}{Martin Fr{\"a}nzle}.}
  \bibinfo{year}{2019}\natexlab{}.
\newblock \showarticletitle{{Robust invariant sets generation for
  state-constrained perturbed polynomial systems}}. In
  \bibinfo{booktitle}{\emph{Proceedings of the 22nd ACM international
  conference on hybrid systems: Computation and control}}.
  \bibinfo{pages}{128--137}.
\newblock


\bibitem[Yasser~Shoukry(2021)]%
        {polyar}
\bibfield{author}{\bibinfo{person}{Wael~Fatnassi Yasser~Shoukry}.}
  \bibinfo{year}{2021}\natexlab{}.
\newblock \showarticletitle{{PolyAR: A Highly Parallelizable Solver For
  Polynomial Inequality Constraints Using Convex Abstraction Refinement}}. In
  \bibinfo{booktitle}{\emph{IFAC-PapersOnLine}}, Vol.~\bibinfo{volume}{54}.
  \bibinfo{pages}{43--48}.
\newblock
Issue 5.
\showISSN{2405-8963}


\bibitem[Zettler and Garloff(1998)]%
        {range3}
\bibfield{author}{\bibinfo{person}{Michael Zettler} {and}
  \bibinfo{person}{J{\"u}rgen Garloff}.} \bibinfo{year}{1998}\natexlab{}.
\newblock \showarticletitle{{Robustness analysis of polynomials with polynomial
  parameter dependency using Bernstein expansion}}.
\newblock \bibinfo{journal}{\emph{IEEE Trans. Automat. Control}}
  \bibinfo{volume}{43}, \bibinfo{number}{3} (\bibinfo{year}{1998}),
  \bibinfo{pages}{425--431}.
\newblock


\end{thebibliography}

\begin{appendix}

\section{Proof of Proposition~\ref{prop:nn_estimate}}
\begin{proof}
Note that $NN_{a_p \rightarrow X_0} (a_p)$ is a vector-valued function which returns the roots $X_0 = (x^1_0, x^2_0, \cdots, x^{n_r}_0)$ of the polynomial $p$. Therefore, to upper bound the Lipschitz constant of $NN_{a_p \rightarrow X_0} (a_p)$, we will start by upper bounding the Lipschitz constant of its components functions $ NN^j_{a_p \rightarrow x^j_0} (a_p)$, $ 1 \leq j \leq n_r$. To that end, consider two polynomials with coefficients $a_p$ and $a'_p$ such that $\norm{a_p - a'_p} \le \epsilon$. Therefore:
\begin{align}
  & \norm{NN^j_{a_p \rightarrow x^j_0} (a_p) - NN^j_{a_p \rightarrow x^j_0} (a'_p)}_2  = \norm{ x^j_0(a_p) - x^j_0(a'_p)}_2 \leq C_{a_p}(x^j_0) \norm{a_p - a'_p} \leq \overline{C}_{a_p} \epsilon  
\end{align}
where $x^j_0(a_p)$ and $x^j_0(a'_p)$ are the location of the $j$th root for the polynomials with coefficients $a_p$ and $a'_p$, respectively. The last two inequalities follow from the definition of the condition number (Definition~\ref{def:condition}). 
Now, 
\begin{align}\label{rootNN}
\norm{NN_{a_p \rightarrow X_0} (a_p) - NN_{a_p \rightarrow X_0} (a'_p)}_2
&= \sqrt{\sum\limits_{j = 1}^{n_r} \norm{NN^j_{a_p \rightarrow x^j_0} (a_p) - NN^j_{a_p \rightarrow x^j_0} (a'_p)}_2^2}
%
%
\\&\leq  \sqrt{n_r \overline{C}_{a_p}^2 \epsilon^2} = \sqrt{n_r} \overline{C}_{a_p} \epsilon.
\end{align}
From which we conclude that the Lipschitz constant of $ NN_{a_p \rightarrow X_0} (a_p)$ is bounded by $\sqrt{n_r} \overline{C}_{a_p} = \mathcal{O}(n_r \overline{C}_{a_p})$.
\end{proof}

\section{Proof of Proposition~\ref{prop:nn_split}}
\begin{proof}
We assume that the number of sub-regions $l$ is fixed for each dimension $n$, and $l^{\frac{1}{n}} = k~\in \mathbb{N}$. 
Partitioning the input space into $l$ sub-regions occurs by dividing the interval for each dimension into $k$ sub-intervals. Without loss of generality, the sub-neural network $NN_{I_n \rightarrow I^{i}_n} (I_n)$ for $n=1$ can be defined as:
\begin{equation}\label{part_eq}
(\underline{d}^i, \overline{d}^i) 
= NN_{I_n \rightarrow I^{i}_n} (I_n) = NN_{I_n \rightarrow I^{i}_n} (\underline{d}, \overline{d})
= \left(\underline{d} + \frac{i}{k} \left(\overline{d} - \underline{d}\right), \; \overline{d} + \frac{i+1}{k} \left(\overline{d} - \underline{d}\right) \right) = 
\begin{bmatrix} 
1+i/k & -i/k \\
-(i+1)/k & 1 + (i+1)/k
\end{bmatrix} 
\begin{bmatrix}
\underline{d} \\ \overline{d}
\end{bmatrix}
\end{equation}
A generalization to a higher dimension is straightforward by replacing $i$ with a multi-index in each dimension. Note that $NN_{I_n \rightarrow I^{i}_n} (I_n)$ is a multivariate linear function in its inputs and hence its Lipschitz constant can be computed as the largest singular value. Indeed, the linear function depends only on the constant $k$ (which depends on the constant $l$ and the dimension $n$) from which we conclude that the Lipschitz constant of $NN_{I_n \rightarrow I^{i}_n}$ is some constant $\mathcal{O}(n)$.

\end{proof}

\section{Proof of Proposition~\ref{prop:nn_indicator}}
\begin{proof}
It follows from equations~\eqref{eq:estimate}-\eqref{eq:indicator} that the sub-neural network $NN_{X_0 \rightarrow ZC_i}$ can be written as:
\begin{align}
    ZC_i(a_p, I^i_n) 
    &= NN_{X_0 \rightarrow ZC_i} \left( NN_{a_p \rightarrow X_0} (a_p), NN_{I_n \rightarrow I^{i}_n} (I_n)\right)
    = NN_{X_0 \rightarrow ZC_i} \left( (x_0^1, \ldots, x_0^{n_r}), (\underline{d}^i, \overline{d}^i) \right).
\end{align} 
The indicator variable $ZC_i$ should be set to zero whenever all the roots $x_0^j$ lies outside the hyperrectangle $I_n^i(\underline{d}^i, \overline{d}^i)$. First note that a root $x_0^j$ lies outside $I_n^i(\underline{d}^i, \overline{d}^i)$ if and only if the following condition holds:
\begin{align}
    x^j_{0} \not\in I_n^i(\underline{d}^i, \overline{d}^i) \quad \Longleftrightarrow \quad
    \sum\limits_{k = 1}^{k = n} \bigg|x^j_{0,k} - \underline{d}^i_k\bigg| +  \bigg| x^j_{0,k} - \overline{d}^i_k\bigg| - \sum\limits_{k = 1}^{n} \left(\overline{d}^i_k - \underline{d}^i_k\right) > 0
\end{align}
where $x^j_{0,k}, \underline{d}^i_k,$ and $\overline{d}^i_k$ are the $k$th element in the vectors $x^j_{0}, \underline{d}^i$ and $\overline{d}^i$, respectively. Hence, the indicator variable $ZC_i$ should be set to zero whenever the following conditions hold:
\begin{align}
    ZC_i(a_p, I^i_n) = 0 \quad \Longleftrightarrow \quad \text{max}_{j \in \{1,\ldots,n_r\}} \left(\sum\limits_{k = 1}^{k = n} \bigg|x^j_{0,k} - \underline{d}^i_k\bigg| +  \bigg| x^j_{0,k} - \overline{d}^i_k\bigg| - \sum\limits_{k = 1}^{n} \left(\overline{d}^i_k - \underline{d}^i_k\right)\right) = 0
    \label{eq:zci_condition}
\end{align}
Before we compute the Lipschitz constant of $NN_{X_0 \rightarrow ZC_i}$ in Equation~\eqref{eq:zci_condition}, we recall the following identities. Consider two functions $f(x)$ and $g(x)$ with Lipschitz constants $L_f$ and $L_g$, respectively. Then:
\begin{itemize}
    \item The Lipschitz constant of $\max(f(x),g(x))$ is bounded by $L_f + L_g$.
    \item  The Lipschitz constant of $f(x) + g(x)$ is bounded by $\max(L_f, L_g)$.
\end{itemize}
Now notice that $|x^j_{0,k} - \underline{d}^i_k| = max(x^j_{0,k} - \underline{d}^i_k,0) + max(-x^j_{0,k} + \underline{d}^i_k,0)$. Applying the identities above along with the fact that the Lipschitz constant of $x^j_{0,k} - \underline{d}^i_k$ is $\mathcal{O}(1)$,  we conclude that the Lipschitz constant of $|x^j_{0,k} - \underline{d}^i_k|$  is $\mathcal{O}(1)$. Hence, the Lipschitz constant of $\sum\limits_{k = 1}^{k = n} \bigg|x^j_{0,k} - \underline{d}^i_k\bigg| +  \bigg| x^j_{0,k} - \overline{d}^i_k\bigg| - \sum\limits_{k = 1}^{n} \left(\overline{d}^i_k - \underline{d}^i_k\right)$ is also $\mathcal{O}(1)$. Finally, the Lipschitz constant of the right hand side of Equation~\eqref{eq:zci_condition} is $\mathcal{O}(n_r)$. We conclude our proof by noticing that all the operators in Equation~\eqref{eq:zci_condition}---namely the absolute value, the max operator, summation, and checking the final value against a constant---can be implemented exactly using ReLU neural networks~\cite{TLL} and hence the neural network $NN_{X_0 \rightarrow ZC_i}$ will also have a Lipschitz constant equal to $\mathcal{O}(n_r)$.

\end{proof}

\section{Proof of Proposition~\ref{prop:nn_count}}
\begin{proof}
This result follows directly by noticing that $NN_{ZC \rightarrow L^{+}/L^{-}}$ can be computed as a linear function:
\begin{align}
    NN_{ZC \rightarrow L^{+}/L^{-}}(ZC_1, \ldots, ZC_l) = v \sum_{i = 1}^l (1 - ZC_i)
\end{align}
where $v$ is a constant that depends on the volume of the hyperrecatngle $I_n$. Since $l$ is a constant, we conclude the result.



\end{proof}

\end{appendix}

\end{document}